\documentclass[reqno,12pt]{article}
\usepackage{palatino}
\usepackage{amssymb,amsthm,amsmath,amsfonts}
\usepackage{epsfig}

\theoremstyle{plain}
\begingroup

\newtheorem{thm}{Theorem}[section]
\newtheorem{theorem}{Theorem}[section]

\newtheorem{lemma}[thm]{Lemma}

\endgroup
\theoremstyle{definition}
\newtheorem{defn}{Definition}[section]
\newtheorem{definition}{Definition}[section]

\newtheorem{remark}[defn]{Remark}

\newtheorem{example}[defn]{Example}
\theoremstyle{remark}

\numberwithin{equation}{section}
\numberwithin{figure}{section}

\newcommand\dee{\partial}

\DeclareMathOperator{\re}{Re} \DeclareMathOperator{\im}{Im}

\def\I{\mathrm{i}}

\def\D{{\mathbb D}}

\def\C{{\mathbb C}}

\begin{document}

\title{
Vortex pairs and dipoles on closed surfaces}
\author{
Bj\"orn Gustafsson\textsuperscript{1},
}

\date{\today}

\maketitle

\tableofcontents

\begin{abstract} 
We set up general equations of motion for point vortex systems on closed Riemannian surfaces, allowing for the case that 
the sum of vorticities is not zero and there hence must be counter-vorticity present.  
The dynamics of global circulations which is coupled to the dynamics of the vortices is carefully taken into account.

Much emphasis is put to the study of vortex pairs, having the Kimura conjecture in focus. This says that vortex pairs 
move, in the dipole limit, along geodesic curves, and proofs for it have previously been given by S.~Boatto and J.~Koiller
by using Gaussian geodesic coordinates. In the present paper we reach the same conclusion by following a slightly
different route, leading directly to the geodesic equation with a reparametrized time variable. 

In a final section we explain how vortex motion in planar domains can be seen as a special case of vortex motion on closed surfaces,
and in two appendices we give some necessary background on affine and projective connections.
 \end{abstract}



\noindent {\it Acknowledgements:} The author has been inspired by discussions with Stefanella Boatto, Jair Koiller, 
Nikolai Nadirashvili, and Jesper G\"oransson.

 \footnotetext[1]
{Department of Mathematics, KTH, 100 44, Stockholm, Sweden.\\
Email: \tt{gbjorn@kth.se}}


\section{Introduction}

This paper extends previous results in \cite{Gustafsson-2019} on multiple point vortex motion on closed Riemannian surfaces of 
arbitrary genus to cases in which the sum of the vorticities is not zero, and therefore
a counter-vorticity must be present. This occurs for example when there is only one point vortex. For reasons
dictated by a Hodge decomposition, the counter-vorticity is naturally chosen to be uniformly distributed on the surface.

Special for our investigations is that we, in the case of higher genus, carefully take the circulation around the holes into account
and make sure that the Kelvin law of preservation of circulations is satisfied. This was done already in \cite{Gustafsson-2019},
but in the case of a counter-vorticity being present one has to take some extra steps. It is natural to choose fixed closed curves on 
the surface to make up a basis for the first homology group, but 
the circulation around such curves will in general not be conserved in time.
In the present paper we handle the so arising problem by considering
these circulations as free variables in the phase space, in addition to the locations of the vortices. If there are $n$ vortices
and the genus of the surface is $\texttt{g}$ the phase space will thus have real dimension $2n+2\texttt{g}$.

The Hamiltonian function will, as usual, be a renormalized kinetic energy for the flow, see (\ref{H}), but in the presence of a counter-vorticity
the symplectic form has to be accordingly adapted, see (\ref{symplectic form}). With this done the dynamical equations come out to be the 
expected ones (Theorem~\ref{thm:dynamics}). 
After having set up the vortex dynamics in general (Section~\ref{sec:hamiltonian} and \ref{sec:hamiltoneq}) 
and after a short discussion of the single vortex case in genus zero (Section~\ref{sec:single vortex}), 
we turn in Section~\ref{sec:vortex pairs} to the question of motion of vortex pairs, i.e. systems 
consisting of two point vortices of equal strength but rotating in opposite direction. In this case there is no counter-vorticity. 
The conjecture  of Y.~Kimura \cite{Kimura-1999}, saying
that a vortex pair in the dipole limit moves along a geodesic curve, has been a major source of inspiration for the present paper. 
The same applies for papers  \cite{Koiller-Boatto-2009, Boatto-Koiller-2013} by S.~Boatto, J.~Koiller.
These latter papers actually contain a proof of Kimura's conjecture, and in the present paper we try to clarify the situation further
by connecting the dipole trajectory directly to the geodesic equation associated to the metric.

Throughout this paper we think of a Riemannian surface (i.e. a Riemannian manifold of dimension two) as a Riemann surface provided with
a metric which is compatible with the conformal structure. Therefore we are lead into more or less classical function theory on Riemann surfaces.
Our treatment is based entirely on local holomorphic coordinates. For example, the logarithmic singularities in the stream function will look like
$\log |z-w|$ near a vortex, $z$ and $w$ being complex coordinates, $w$ representing the location of a vortex. When such a singular term is removed 
from a local series expansion (with respect to $z$) around $w$ in the necessary renormalization process and $z$ is set equal to $w$, 
then the quantity which remains behaves no longer  as a function under changes of coordinates.
Instead it is a kind of ``connection''.
For such reasons we run into differential geometric considerations involving different kinds of connections: affine connections 
(or $1$-connections) and projective, or Schwarzian, connections ($2$-connections). The metric itself can be viewed as the exponential
of a $0$-connection. The necessary definitions and background material are summarized in two
appendices, Sections~\ref{sec:connections} and \ref{sec:singular parts}. 
Affine connections appear in many areas of mathematical physics, and as a
mathematical tool they show up in expressions for covariant derivatives, and 
in particular in the equation (\ref{imgeodesics}) for geodesic curves. Also projective connections are regularly used, for example in
conformal field theory. 

The Kimura conjecture and other matters related to dipole motion has been discussed also in  
\cite{Llewellyn-2011, Llewellyn-Nagem-2013, Koiller-2018, Koiller-2020, Krishnamurthy+-2021}, and from slightly different points of view in
 \cite{Chorin-1973, Kulik-Tur-Yanovsky-2010, Holm+-2017, Cawte+-2019}.
As mentioned, the analysis in the present paper leads to a new confirmation of Kimura's conjecture, although in a rather weak form:
expanding all quantities in Taylor series and keeping only the leading terms, the dynamical equations for a vortex pair
reduce to the geodesic equation in the dipole limit. See Theorem~\ref{thm:kimura}.  It seems likely that stronger forms of the conjecture
(convergence of the trajectories, for example) can be formulated and proved, but we leave such matters for possible future investigations.

It may be remarked that Kimura's conjecture is counter-intuitive. The reason to think so is that vortex motion in general is governed by
global laws on the manifold, like the structure of Green functions and other harmonic functions, whereas the geometry of
geodesics is an entirely local matter. If the geometry of the manifold is changed at one place then this will not affect what geodesics
look like at another place. It will however change the structure of Green functions and the general vortex dynamics everywhere. 
The solution of this paradox is that
vortex dipoles are highly singular. A vortex pair is partly governed by the global harmonic structure, but in the limit, when the
vortex pair becomes a dipole, this harmonic structure is completely overruled by the differential geometric structure.  In that limit,
all terms in the dynamical equations containing harmonic functions become negligible compared to those terms depending
on the metric only. See Section~\ref{sec:vortex pairs}. 

In Section~\ref{sec:dipoles} we make some attempts to treat vortex dipoles directly by starting from the dynamical equation for a single
vortex and taking the distributional derivative with respect to the coordinate $w$ which gives the location of the vortex. 
The local stream function for the flow near a single vortex at $w$ is something like (we ignore constant factors below)
$$
\psi(z)=\log |z-w|+ \text{regular terms},
$$
and the corresponding flow, when considered as a one-form (or covariant vector field) looks like 
$$
\nu(z)=\frac{dz}{z-w} + \text{complex conjugate} + \text{regular terms}.
$$
The term named ``complex conjugate'' is inserted to make possible for $\nu$ to be real-valued.
For the dipole we accordingly have
$$
d_w\nu(z)=\frac{dz dw}{(z-w)^2} +\text{complex conjugates} + \text{regular terms}.
$$
The differential $dw$ shall then be thought of as containing information of the orientation of the dipole.

The constant term (i.e. the first term in the Taylor series of the regular terms) are in the above three expressions a $0$-connection, a $1$-connection, and
a $2$-connection, respectively (again up to constant factors). See Lem\-ma~\ref{lem:connections} for a precise statement. The $0$-connection for the 
stream function is the Routh's stream function, or Robin function when $\psi$ is thought of as a Green function. And the $1$-connection for $\nu$ is
what gives the speed of the vortex, after subtraction (in case of a curved manifold) of the affine connection coming from the metric. The result
will then be a covariant vector field.

For $d_w\nu(z)$ the constant term is similarly a $2$-connection, or projective connection, but it is more unclear what influence it has on the 
motion of the dipole. Indeed, the dipole singularity is very strong and it has a definite direction at $w$, and it seems that all finite terms should
be negligible in comparison.   Therefore it is difficult to figure out any dynamics from this picture, only that the dipole should move with
infinite speed in the direction dictated by the singularity, presumably (and as have been confirmed) along a geodesic with respect to the metric.

In the last section (not counting the appendices) of the paper we discuss briefly how point vortex motion in planar domains can be considered as a special case of
vortex motion on surfaces. This is done by doubling the domain to a compact Riemann surface (the Schottky double), which is a standard
tool as far as the conformal structure is concerned. What is special in our case is that we have to take the
metric structure into account, and this becomes non-smooth in the doubling procedure. An interesting observation
is that the boundary of the domain becomes a geodesic curve with respect to the natural planar metric (on each of the halves)
of the Schottky double.  

In general the present paper is, in addition to the papers by S.~Boatto and J.~Koiller already mentioned,
much in spirit of papers \cite{Llewellyn-2011, Llewellyn-Nagem-2013, Grotta-2017, Grotta-Ragazzo-Barros-Viglioni-2017, Wang-2021}
by S.~G.~Llewellyn Smith, R.~J.~Nagem, C.~Grotta Ragazzo, H.~Vig\-lioni, and Q.~Wang.
The paper \cite{Bogatskiy-2019}  by A.~Bogatskiy
contains ideas concerning the higher genus case which are somewhat similar to those in the present paper, 
see in particular Section~A.3 in \cite{Bogatskiy-2019}.
We wish to mention also the work \cite{Borisov-Mamaev-2006} by A.~V.~Borisov and I.~S.~Mamaev, and the very early article 
\cite{Fridman-Polubarinova-1928} (cited in \cite{Borisov-Mamaev-2006, Llewellyn-Nagem-2013}) by A.~A.~Fridman and
P.~Ya.~Polubarinova. That work, from 1928, may be one of the first papers on motion of dipoles and other higher singularities.
Fridman was a famous cosmologist and Pelageya Yakolevna Polubarinova (or Polubarinova-Kochina, after marriage) was a young 
(at that time) disciple of Fridman. This remarkable woman wrote her first paper in 1924, and was scientifically active
throughout her life. Her published papers and books span a period of 75 years, the last paper appearing the same 
year as she died, one hundred years of age. See \cite{Vasiliev-2009} for a short biography related to her work on Hele-Shaw flows.


\section{Fluid dynamics on a Riemannian surface}\label{sec:general}

\subsection{General notations and assumptions}

We consider the dynamics of a non-viscous incompressible fluid on a compact Riemannian manifold
of dimension two. We follow the treatments in \cite{Schutz-1980, Arnold-Khesin-1998, Frankel-2012} 
in the respects of treating the fluid velocity field as a one-form and in extensively using the Lie derivative.
These sources also provide the standard notions and notations for differential geometry to be used.
Other treatises in fluid dynamics,  suitable for our purposes, are \cite{Marchioro-Pulvirenti-1994, Newton-2001}.

Traditionally, non-viscous fluid dynamics is discussed in terms of the fluid velocity vector field ${\bf v}$, the density $\rho$, and the pressure $p$ of the fluid.
The basic equations are the equation of continuity (expressing conservation of mass) and Euler's equation (conservation of momentum):
\begin{equation}
\frac{\partial \rho}{\partial t}+\nabla\cdot (\rho{\bf v})=0,
\label{Euler1}
\end{equation}
\begin{equation}
\frac{\partial {\bf v}}{\partial t}+({\bf{v}}\cdot\nabla) {\bf v}=-\frac{1}{\rho}\nabla p.
\label{Euler2}
\end{equation}
These are to be combined with a constitutive law giving a relationship between $p$ and $\rho$. 

The above equations are already simplified versions of the general Navier-Stokes equations (which allow for viscosity terms), 
but we shall simplify further by working only in two dimensions and by taking the constitutive law to be the simplest possible:
\begin{equation}\label{incompressibility}
\rho=1.
\end{equation}
Thus $\rho$ disappears from discussion, and the equation of continuity becomes
\begin{equation}\label{nablau}
\nabla\cdot {\bf v}=0.
\end{equation}

One can get rid also of $p$, because when $\rho$ is constant
then $p$ appears only as a scalar potential in Euler's equation, this equation effectively saying that
\begin{equation}
\frac{\partial {\bf v}}{\partial t}+({\bf{v}}\cdot\nabla) {\bf v}=\nabla ({\rm something}),
\label{euler1}
\end{equation}
where this ${\rm ``something"}=p$\, afterwards can be recovered, up to a (time dependent) constant.

So everything is extremely simple (in theory), and even more so in two dimensions. As is well-known \cite{Schiffer-Spencer-1954},
every (oriented) Riemannian manifold of dimensions two can be made into a Riemann surface
by choosing isothermal local coordinates $(x,y)=(x^1,x^2)$, by which $z=x+\I y$ becomes a holomorphic coordinate
and the metric takes the form
\begin{equation}
ds^2 =\lambda(z)^2 |dz|^2=\lambda (z)^2(dx^2+dy^2)
\label{metric}
\end{equation} 
with $\lambda>0$. Thus the metric tensor $g_{ij}$, as appearing in general in $ds^2=g_{ij}dx^i dx^j$, becomes
$g_{ij}=\lambda^2\delta_{ij}$ in these coordinates. It turns out to be convenient to work with the fluid velocity $1$-form
(or covariant vector)
$$
\nu= \nu_x \,dx+\nu_y\,dy
$$ 
rather than with the corresponding vector field, which then becomes
$$
{\bf v}=\frac{1}{\lambda^2} (\nu_x \,\frac{\partial}{\partial x}+\nu_y \,\frac{\partial}{\partial y}).
$$

The Hodge star operator takes $0$-forms into $2$-forms and vice versa, and takes $1$-forms to $1$-forms. 
It acts on basic differential forms as
$$
*1=\lambda^2 (dx\wedge dy)= {\rm vol}=\text{the volume (area) form},
$$
$$
*dx=dy, \quad *dy=-dx,\quad *(dx\wedge dy)=\lambda^{-2}.
$$
Thus, for $1$-forms,
$$
*\nu =-\nu_y\,dx +\nu_x\,dy,
$$
which can be interpreted as a rotation by ninety degrees of the corresponding vector.

In addition to the Hodge star it is useful to introduce the Lie derivative $\mathcal{L}_{\bf v}$ of a vector field ${\bf v}$.
When acting on differential forms this is related to interior derivation (or ``contraction'') $i({\bf v})$ 
and exterior derivation $d$ by the homotopy formula
\begin{equation}\label{cartan}
\mathcal{L}_{\bf{v}}=d\circ i({\bf{v}})+  i({\bf{v}})\circ d.
\end{equation}
The Hodge star and interior derivation are related by
\begin{equation}\label{Hodgei}
i({\bf v}) {\rm vol}=*\nu,
\end{equation}
where ${\bf v}$ and $\nu$ are linked via the metric tensor as above. Thus
$$
d*\nu=d(i({\bf v}){\rm vol})=\mathcal{L}_{\bf v}({\rm vol})= (\nabla\cdot {\bf v}) {\rm vol}.
$$
See \cite{Frankel-2012} for this identity, 
and for further details in general. We conclude that (\ref{nablau}) is equivalent to the statement that $*\nu$ is a closed form:
\begin{equation}\label{dstarnu}
d*\nu=0.
\end{equation} 
Locally we can therefore write
\begin{equation}
*\nu=d\psi
\label{nudpsi}
\end{equation}
for some (locally defined) {\it stream function} $\psi$.
The {\it vorticity} $2$-{\it form} for a fluid is, in terms of the flow $1$-form $\nu,$
\begin{equation}\label{omegadnu}
\omega=d\nu= (\frac{\partial \nu_y}{\partial x}-\frac{\partial \nu_x}{\partial y}) \,dx\wedge dy.
\end{equation}

The Euler equation (\ref{Euler2}) can be written
\begin{equation}\label{euler}
(\frac{\partial }{\partial t}+\mathcal{L}_{\bf v})({\bf \nu})=d(\frac{1}{2}|{\bf v}|^2-{p}).
\end{equation}
Note that the left member involves the fluid velocity both as a vector field and as a form. The combination
$\frac{\partial }{\partial t}+\mathcal{L}_{\bf v}$
can be viewed as a counterpart, for forms, to the more traditional ``convective derivative'' (or material derivative) 
$\frac{\partial }{\partial t}+{\bf{v}}\cdot\nabla$, implicitly used in (\ref{euler1}), 
for vector fields. However, the two derivatives are not the same.
The equivalence between (\ref{euler}) and (\ref{Euler2}) (when $\rho=1$) is a consequence of the identity
$$
\mathcal{L}_{\bf v}\nu=d(\frac{1}{2}|{\bf v}|^2)+({\bf{v}}\cdot\nabla) {\nu},
$$
which can be directly verified by coordinate computations. See \cite{Schutz-1980, Frankel-2012} for details.

Since the pressure $p$ does not appear in any other equation, 
the Euler equation on the form (\ref{euler}) only expresses, in analogy with (\ref{euler1}), that the left member is an exact $1$-form:
\begin{equation}\label{eulerinvariant}
(\frac{\partial }{\partial t}+\mathcal{L}_{\bf v})({\bf \nu})={\rm exact}=d\phi.
\end{equation}
From the scalar $\phi$ the pressure $p$ then can be recovered via
\begin{equation}\label{phi}
\phi=\frac{1}{2}|{\bf v}|^2-{p}+{\rm constant}.
\end{equation}

On acting by $d$ on (\ref{eulerinvariant}),
the local form of Helmholtz-Kirchhoff-Kelvin law of conservation of vorticity follows:
\begin{equation}\label{helmholtz}
(\frac{\partial}{\partial t}+\mathcal{L}_{\bf{v}})(\omega)=0.
\end{equation}
Since the right member of (\ref{eulerinvariant}) is not only a closed differential form, but even an exact form,
a stronger, global, version of the Helmholtz-Kirchhoff-Kelvin law actually follows. This can be expressed by saying that
\begin{equation*}\label{ddtgamma}
\frac{d}{dt} \oint_{\gamma(t)} \nu =0 
\end{equation*}
for any closed curve $\gamma(t)$ which moves with the fluid.


\section{Green functions and harmonic forms}\label{sec:Hodge}

\subsection{The one point Green function}\label{sec:one point green}

In the sequel $M$ will be a closed (compact) Riemann surface provided with a Riemannian  metric on the form (\ref{metric}).

Given a $2$-form $\omega$ on $M$ one can immediately obtain a corresponding  
Green potential $G^\omega$ by Hodge decomposition, i.e. by orthogonal decomposition in the Hilbert space of square integrable $2$-forms.
The inner product for such forms is
\begin{equation}\label{energy2}
(\omega_1,\omega_2)_2=\int_M \omega_1 \wedge *\omega_2.
\end{equation}
The given $2$-form $\omega$ then decomposes as the orthogonal decomposition of an exact form and a harmonic form:
$$
\omega =d (\rm something)+ {\rm harmonic}.
$$
This Hodge decomposition can more precisely be written as
\begin{equation} 
\omega =-d*dG^\omega + {\rm constant} \cdot {\rm vol},
\label{hodge}
\end{equation}
where $G^\omega$ is normalized to be orthogonal to all harmonic $2$-forms, namely satisfying
\begin{equation}
\int_M G^\omega \,{\rm vol}=0.
\label{tvol}
\end{equation}
The constant in (\ref{hodge}) necessarily equals the mean value of $\omega$,
\begin{equation}\label{constant}
{\rm constant} = \frac{1}{V}\int_M\omega, 
\end{equation}
where $V$ denotes the total volume (=area) of $M$:
$$
V=\int_M {\rm vol}.
$$

When $\omega$ is exact, as in (\ref{omegadnu}), the second term in  the right member of (\ref{hodge}) disappears, hence
\begin{equation}\label{ddGomega}
-d*dG^\omega =\omega
\end{equation}
in this case. In the other extreme, when $\omega$ is harmonic, i.e. is a constant multiple of ${\rm vol}$,
the first term disappears. Indeed, the whole Green function disappears:
\begin{equation}\label{Gvol}
G^{{\rm vol}}=0.
\end{equation}

For $1$-forms the inner product has the same expression as for $2$-forms:
\begin{equation}\label{energy1}
(\nu_1,\nu_2)_1=\int_M \nu_1 \wedge *\nu_2.
\end{equation}
If $\nu$ represents a fluid velocity, then $(\nu, \nu)_1$ has the interpretation of being (twice) the {kinetic energy}
of the flow. For a function (potential) $u$ we consider the Dirichlet integral $(du,du)_1$ to be its energy.
Thus constant functions have no energy. 

The {energy}  $\mathcal{E}(\omega,\omega)$ of any $2$-form $\omega$ is 
defined to be the energy  of its Green potential $G^\omega$.
Thus, for the {mutual energy}, 
\begin{equation}\label{Eomega}
\mathcal{E}(\omega_1, \omega_2)=(dG^{\omega_1}, dG^{\omega_2})_1
=\int_M G^{\omega_1} \wedge \omega_2.
\end{equation}
It follows, from (\ref{tvol}) for example, that the volume form has no energy:
$$
\mathcal{E}({\rm vol}, {\rm vol})=0.
$$

It was tacitly assumed above that the forms under discussion belong to the $L^2$-space defined by the inner product. 
However, the mutual energy sometimes extends to circumstances in which source distributions 
have infinite energy. This applies in particular to the Dirac current $\delta_a$, which we consider as a $2$-form with 
distributional coefficient, namely defined by the property that
$$
\int_M \delta_a \wedge \varphi =\varphi(a)
$$ 
for every smooth function $\varphi$.
Certainly $\delta_a$ has infinite energy, but if $a\ne b$, then $\mathcal{E}(\delta_a, \delta_b)$ is still finite and has a natural interpretation:
it is the {\it (one-point) Green function}, or ``monopole'' Green function: 
\begin{equation}\label{Gab}
G(a,b)=G^{\delta_a}(b)= \mathcal{E}(\delta_a, \delta_b).
\end{equation}
Here the first equality can be taken as a definition, and then the second equality
follows on using (\ref{hodge}), (\ref{tvol}):
\begin{align*}
  \mathcal{E}(\delta_a, \delta_b) &= \int_M dG^{\delta_a}\wedge *dG^{\delta_b}= -\int_M G^{\delta_a}\wedge d*dG^{\delta_b} \\
                                  &= \int_M G^{\delta_a}\wedge \left( \delta_b-\frac{1}{V}\,{\rm vol} \right) = G^{\delta_a}(b) = G(a,b).
\end{align*}
This shows in addition that $G(a,b)$ is symmetric. 

Changing now notations from $a,b$ to $z,w$, where later $w$ will have the role of being the location of a point vortex, the Green function $G(z,w)$, as a function
of $z$, has just one pole (at $z=w$). Its Laplacian, as a $2$-form, is then
\begin{equation}\label{ddGdeltavol}
-d* dG(\cdot,w)=\delta_w- \frac{1}{V}{\rm {vol}}.
\end{equation}
It is interesting to notice that among the two terms in the right member of (\ref{ddGdeltavol}), one has infinite energy 
and one has zero energy ($\mathcal{E}(\delta_w, \delta_w)=+\infty$, $\mathcal{E}({\rm vol},{\rm vol})=0$). 

\begin{remark}
It is more common to treat the Dirac delta and the Laplacian (denoted $\Delta$) as ``densities'' with respect to the volume form. 
However, we find our usage convenient.  In any case, the relationships are
$$
\delta_a = ({\rm delta\,\, ''function''}) \,{\rm vol}, \quad d*d \varphi =(\Delta \varphi) \, {\rm vol}.
$$ 
\end{remark}

\begin{remark}[Notational remark]\label{rem:notation} 
Letters like $z$, $w$ will have the double roles of being complex-valued local coordinates
on parts of the Riemann surface and of denoting the points on the surface for which the coordinates have the values in question.
A more formal treatment could use, for example, $P\in M$ for a point and $z(P)\in \C$ for the corresponding coordinate value. 
\end{remark}

\begin{remark}[Real versus complex notation]\label{rem:real-complex}
For future needs we wish to clarify the relationship between real and complex coordinates in the context of tangent and cotangent vectors.

Let $z=x+\I y$ be a complex coordinate on $M$ and consider a curve $t\mapsto z(t)$ in $M$ (for example the trajectory of a
vortex). Set $\dot{z}={dz}/{dt}$, and similarly for $\dot{x}$ and $\dot{y}$, so that $\dot{z}=\dot{x}+\I\dot{y}$. The velocity of this
moving point is primarily to be considered as a vector in the (real) tangent space of $M$ (at the point under consideration). This gives, 
in the picture of viewing tangent vectors as derivations, the velocity vector
\begin{equation}\label{vectorV}
{\bf V}=\dot{x}\frac{\partial}{\partial x}+\dot{y}\frac{\partial}{\partial y}=\dot{z}\frac{\partial}{\partial z}+\dot{\bar{z}}\frac{\partial}{\partial \bar{z}},
\end{equation}
where it is understood that $\dot{\bar{z}}=\overline{\dot{z}}$.

The real tangent space used above can in a next step be complexified, which means that one breaks the connection between 
$\dot{z}$ and $\dot{\bar{z}}$ and consider them as independent complex variables. Equivalently, one allows $\dot{x}$ and $\dot{y}$
to be complex-valued. The so obtained complex tangent space
can be decomposed as a direct sum of its holomorphic and anti-holomorphic subspaces, and then
$$
{\rm proj}:\quad
{\bf V}=\dot{z}\frac{\partial}{\partial z}+\dot{\bar{z}}\frac{\partial}{\partial \bar{z}}\mapsto \dot{z}\frac{\partial}{\partial z}
$$
is a natural and useful identification of the real tangent space with the holomorphic part of the complex tangent space. See \cite{Griffiths-Harris-1978},
in particular Section~2 of Chapter~0, for further discussions.  With this identification $\dot{z}$ 
represents the velocity of $z(t)$. Still one need to keep in mind that $\dot{z}$ is just a complex number and that it is rather
the preimage under ${\rm proj}$ above that is the true velocity, as a real tangent vector. 

The vector ${\bf V}$ above
correspond, via the metric, to the covector
\begin{equation}\label{covectorV}
{\lambda^2}(\dot{x}dx+\dot{y}dy)=\frac{\lambda^2}{2}(\dot{\bar{z}}dz+\dot{{z}}d\bar{z}),
\end{equation}
where $\lambda=\lambda(x,y)=\lambda(z)$ (depending on context). Note that
$$
\frac{1}{2}\dot{\bar{z}}=\frac{1}{2}(\dot{x}-\I \dot{y}), \quad\frac{1}{2}\dot{{z}}=\frac{1}{2}(\dot{x}+\I \dot{y}),
$$
as coefficients of $dz$ and $d\bar{z}$ (respectively)
have similar roles (and signs) as the Wirtinger derivatives $\partial/\partial z$, $\partial/\partial \bar{z}$.
For a covector $\nu$ in general we therefore define
$$
\nu_z=\frac{1}{2}(\nu_x-\I \nu_y), \quad \nu_{\bar{z}}=\frac{1}{2}(\nu_x+\I\nu_y),
$$
so that
$$
\nu=\nu_x dx+\nu_y dy =\nu_zdz+\nu_{\bar{z}}d\bar{z}.
$$
The corresponding (contravariant) vector is then, as in our fluid dynamical contexts,
$$
{\bf v}=\frac{1}{\lambda^2}\big(\nu_x\frac{\partial}{\partial x}+\nu_y\frac{\partial}{\partial y}\big)
=\frac{1}{\lambda^2}\big(\nu_z\frac{\partial}{\partial z}+\nu_{\bar{z}}\frac{\partial}{\partial \bar{z}}\big).
$$
\end{remark}


\subsection{Harmonic one-forms and period relations}

For later use we record the formulas (differentiation with respect to $z$)
$$
dG(z,w)=\frac{\partial G(z,w)}{\partial z}dz+\frac{\partial G(z,w)}{\partial \bar{z}}d\bar{z}=2\re \big(\frac{\partial G(z,w)}{\partial z}dz\big),\qquad 
$$
\begin{equation}\label{stardG0}
*dG(z,w)=-\I\frac{\partial G(z,w)}{\partial z}dz+\I\frac{\partial G(z,w)}{\partial \bar{z}}d\bar{z}=2\im \big(\frac{\partial G(z,w)}{\partial z}dz\big).
\end{equation}
If $\gamma$ is any closed oriented curve in $M$ then clearly $\oint_\gamma dG(\cdot,w)=0$, while the conjugate period defines 
a function 
\begin{equation}\label{Ugamma}
U_\gamma(w)=\oint_\gamma *dG(\cdot,w)=\oint_\gamma \big(-\I\frac{\partial G(z,w)}{\partial z}dz+\I\frac{\partial G(z,w)}{\partial \bar{z}}d\bar{z}\big)
\end{equation}
\begin{equation}\label{dGdGbar}
=-2\I\oint_\gamma \frac{\partial G(z,w)}{\partial z}dz=2\I\oint_\gamma\frac{\partial G(z,w)}{\partial \bar{z}}d\bar{z},
\end{equation}
which, away from $\gamma$, is harmonic in $w$ and makes a unit additive jump as $w$ crosses $\gamma$ from the left to the right.
The harmonicity of $U_\gamma(w$) is perhaps not obvious from outset since $G(\cdot,w)$ is not itself harmonic, but the deviation from harmonicity, namely the extra term in (\ref{ddGdeltavol}), is independent of $w$ and therefore disappears under differentiation. 

Differentiating (\ref{Ugamma}) with respect to $w$ gives 
$$
dU_\gamma(w)=-2\I\oint_\gamma\frac{\partial^2 G(z,w)}{\partial {z}\partial w}d{z}\otimes dw
-2\I\oint_\gamma\frac{\partial^2 G(z,w)}{\partial {z}\partial \bar{w}}d{z}\otimes d\bar{w},
$$
$$
*dU_\gamma(w)=-2\oint_\gamma\frac{\partial^2 G(z,w)}{\partial {z}\partial w}d{z}\otimes dw
+2\oint_\gamma\frac{\partial^2 G(z,w)}{\partial {z}\partial \bar{w}}d{z}\otimes d\bar{w}  \quad
$$
(integration with respect ot $z$, Hodge star and $d$ with respect to $w$).
We have written out a tensor product sign to emphasize that the product between the $dz$ and $dw$ is not a wedge product. 
Adding and subtracting  we obtain the analytic (respectively anti-analytic) differentials
\begin{equation}\label{dUstardU}
dU_\gamma(w) +\I*dU_\gamma(w)=2\frac{\partial U_\gamma(w)}{\partial w}dw=-4\I\oint_\gamma\frac{\partial^2 G(z,w)}{\partial {z}\partial w}d{z}\otimes dw,
\end{equation}
\begin{equation}\label{dUstardUminus} 
dU_\gamma (w)-\I*dU_\gamma(w)=2\frac{\partial U_\gamma(w)}{\partial \bar{w}}d\bar{w}=-4\I\oint_\gamma\frac{\partial^2 G(z,w)}{\partial {z}\partial \bar{w}}d{z}\otimes d\bar{w}.
\end{equation}
 
Let now $\{\alpha_1, \dots, \alpha_\texttt{g}, \beta_1,\dots, \beta_\texttt{g}\}$ be representing cycles for a canonical homology basis for $M$ 
such that each $\beta_j$ intersects $\alpha_j$ once from the right to the left and no other intersections occur (see \cite{Farkas-Kra-1992} for details).
Then there are corresponding harmonic differentials $dU_{\alpha_j}$, $dU_{\beta_j}$ obtained on choosing $\gamma=\alpha_j,\beta_j$ 
in the above construction, and these constitute a basis of harmonic differentials associated with the chosen homology basis.
Precisely we have, for $k,j=1,\dots,\texttt{g}$,
\begin{equation}\label{taualpha}
\oint_{\alpha_k} (-dU_{\beta_j})=\delta_{kj}, \quad \oint_{\beta_k} (-dU_{\beta_j})=0,
\end{equation}
\begin{equation}\label{taubeta}
\oint_{\alpha_k} dU_{\alpha_j}=0, \qquad \oint_{\beta_k} dU_{\alpha_j}=\delta_{kj}. 
\end{equation}
The Kronecker deltas, $\delta_{kj}$, come from the discontinuity (jump) properties of $U_\gamma$ mentioned after (\ref{Ugamma}). 

The integrals of the basic harmonic differentials along non-closed curves become periods of two point Green functions:
\begin{equation}\label{ointdVa}
\int_b^a dU_{\alpha_j}=\oint_{\alpha_j} *dG^{\delta_a-\delta_b},
\end{equation}
\begin{equation}\label{ointdVb} 
\int_b^a dU_{\beta_j}=\oint_{\beta_j} *dG^{\delta_a-\delta_b}.
\end{equation}
Here the integration in the left member is to be performed along a path that does not intersect the curve in the right member.
These formulas are immediate from the definition (\ref{Ugamma}) of $U_\gamma$. 


\section{Energy renormalization and the Hamiltonian}\label{sec:hamiltonian}

\subsection{The renormalized kinetic energy}

When $\nu$ is the flow $1$-form of an incompressible fluid, the equation of continuity says that $d*\nu=0$
(see (\ref{dstarnu})). The  vorticity $2$-form is  $\omega=d\nu$, and (\ref{ddGomega}) holds.
Thus $d(\nu+*dG^\omega)=0$, and setting
\begin{equation}\label{nudG}
\eta=\nu +*dG^\omega
\end{equation}
it follows that $\eta$ is harmonic: $d\eta=0=d*\eta$. Locally we can write $*\eta=d\psi_0$ for some 
harmonic function $\psi_0$. Then
\begin{equation}\label{Ghydro} 
\psi=G^\omega+\psi_0
\end{equation}
becomes a locally defined stream function for the flow, so that 
$$
*\nu=d\psi =dG^\omega +d\psi_0.
$$ 
Different local choices of $\psi$ may differ by additive time dependent constants.
The roles of the additional terms $\eta$ and $\psi_0$, complementing the contribution from the Green function, 
will be somewhat clarified in the context of planar vortex motion and the hydrodynamic
Green function in Section~\ref{sec:Schottky double}.

The relation (\ref{nudG}) written on the form  $\nu=\eta-*dG^\omega$ is an orthogonal decomposition 
of the flow $1$-form $\nu$ with respect to the inner product (\ref{energy1}). Indeed,
$$
(\eta, -*dG^\omega)_1 =\int_M \eta \wedge dG^\omega = -\int_M d\eta\wedge G^\omega=0,
$$ 
since $\eta$ is harmonic. It follows that (twice) the {\it total (kinetic) energy} $(\nu,\nu)_1$ of the flow is given by
$$
2(\nu,\nu)_1=\int_M \nu\wedge *\nu=\frac{\I}{2}\int_M (\nu+\I*\nu)\wedge (\nu-\I*\nu)=
$$
$$
=\int_M dG^\omega \wedge *dG^\omega + \int_M \eta\wedge*\eta
=\mathcal{E}(\omega,\omega)+ \int_M \eta\wedge*\eta=
$$
$$
=\int_M G^\omega\wedge \omega +\sum_{j=1}^{\texttt g}(\oint_{\alpha_j}\eta\oint_{\beta_j}*\eta-\oint_{\alpha_j}*\eta\oint_{\beta_j}\eta).
$$

We shall now specialize on the point vortex case, with vortices of strengths $\Gamma_k$ located at points
$w_k\in M$ ($k=1,\dots, n$). However, we shall not assume that these strengths add up to zero. Hence the sum
\begin{equation}\label{totalGamma}
\Gamma=\sum_{k=1}^n \Gamma_k
\end{equation}
may be non-zero and there will then be a compensating uniform counter vorticity. 
The total vorticity $\omega=d\nu$, which satisfies $\int_M \omega=0$, appears as the right member in
\begin{equation}\label{ddGGamma}
-d*dG^\omega =\omega=\sum_{k=1}^n \Gamma_k \delta_{w_k}- \frac{\Gamma}{V} {\rm vol}.
\end{equation} 
Explicitly we have, on recalling (\ref{Gvol}),
\begin{equation}\label{GGammaG}
G^{\omega} (z) =\sum_{k=1}^n \Gamma_k { G(z,w_k)}.
\end{equation}

On using (\ref{stardG0}) the conjugate $\alpha$-periods of $G^\omega$ can be expressed as
\begin{equation}\label{intalphajdG}
\oint_{\alpha_j}*d{G^\omega}
= -\sum_{k=1}^n2\I\Gamma_k\oint_{\alpha_j}\frac{\partial G(z,w_k)}{\partial z}dz.
\end{equation}
By differentiation and using also (\ref{dGdGbar}), (\ref{dUstardU}) we have 
\begin{equation}\label{ddwkint}
\frac{\partial}{\partial w_k}\oint_{\alpha_j}*d{G^\omega}=\Gamma_k\frac{\partial U_{\alpha_j}(w_k)}{\partial w_k}
=\I\Gamma_k\frac{\partial U^*_{\alpha_j}(w_k)}{\partial w_k},
\end{equation}
\begin{equation}\label{ddwkintalphaj}
\frac{\partial}{\partial w_k}\big(\oint_{\alpha_j}*d{G^\omega}\big)dw_k
=\frac{\Gamma_k}{2}(dU_{\alpha_j}(w_k)+\I*dU_{\alpha_j}(w_k)).
\end{equation}
Here $U^*_{\alpha_j}$ denotes a harmonic conjugate of $U_{\alpha_j}$, whereby $*dU_{\alpha_j}=d(U^*_{\alpha_j})$.
Similar relations hold for periods around $\beta_j$, and for $\bar{w}_k$ derivatives. 

The expression for the kinetic energy can in the point vortex case be written, at least formally, 
$$
2(\nu,\nu)_1=\sum_{k,j=1}^n \Gamma_k \Gamma_jG(w_k,w_j) 
+\sum_{j=1}^{\texttt g}\big(\oint_{\alpha_j}\eta\oint_{\beta_j}*\eta-\oint_{\alpha_j}*\eta\oint_{\beta_j}\eta\big).
$$
However, the presence of the terms $\Gamma_k^2G(w_k,w_k)$ makes the first term become infinite.
In order to renormalize this singular behavior we isolate the logarithmic pole in the Green function by writing
\begin{equation}\label{GlogH}
G(z, w)=\frac{1}{2\pi} (-\log |z-w| +H(z,w)),
\end{equation}
and expand, for a fixed $w$, the regular part in a power series in $z$ as
\begin{equation}\label{greentaylor}
H(z, w)= h_0 (w)+\frac{1}{2}\left(h_1(w)(z-w)+ \overline{h_1 (w)} (\bar{z}-\bar{w})\right)+
\end{equation}
$$
+\frac{1}{2}\left(h_{2}(w)(z-w)^2+\overline{h_{2}(w)}(\bar{z}-\bar{w})^2\right)+h_{11}(w)(z-w)(\bar{z}-\bar{w})+ \mathcal{O}(|z-w|^3).
$$

In Appendix~2, Section~\ref{sec:singular parts}, it is discussed how the coefficients in the expansion (\ref{greentaylor}) behave under
conformal mapping. For example, the coefficient $h_{11}(w)$ transforms as the density of a metric, and it is 
indeed proportional to the given metric:
\begin{equation}\label{h11lambda}
h_{11}(w)=\frac{\pi }{2V}\lambda(w)^2.
\end{equation}
See (\ref{Hlambda}). 
The coefficients $h_0(w)$, $h_1(w)$ and $h_2(w)$ transform, in certain combinations, as ``connections'' under conformal mappings.
For $h_0(w)$ this amounts to saying that it too defines  a metric, namely the {\it Robin metric}, via
\begin{equation}\label{robinmetric}
ds=e^{-h_0(w)}|dw|.
\end{equation}
Metrics of this kind are implicit in the theory of capacity functions, as exposed in \cite{Sario-Oikawa-1969}.
It should be pointed out that our version of the Robin metric depends on the given metric, since the Green
function itself depends on it. The Robin metric can also be adapted to various given circulations, and then it
becomes more intrinsically hydrodynamic in nature.
The coefficient $h_0(w)$ is one example of a {\it (coordinate) Robin function}.

Now, letting $w$ have the role of being a vortex point
indicates that one could renormalize the kinetic energy by simply discarding the singular term $\log|z-w|$, as this seems at first sight 
to produce just a circular symmetric flow, not affecting the speed of the vortex.  However, this is not fully correct in the case of curved surfaces. 
The term  $\log|z-w|$ cannot be just removed, it need be replaced by a term which
counteracts the inhomogeneous transformation law of $h_0(w)$ (see (\ref{h0})).  Such a term comes naturally from the given metric 
$ds=\lambda(w)|dw|=e^{\log\lambda(w)}|dw|$. We see that minus $\log \lambda (w)$ has the right properties, 
and it  combines with $h_0(w)$ into
\begin{equation}\label{Rrobin}
R_{\rm robin}(w)=\frac{1}{2\pi}(h_0(w)+\log\lambda(w)).
\end{equation}
This is indeed a function, the {\it Robin function}, and it appears naturally when writing the singularity of the Green function in terms of
the distance with respect to the given Riemannian metric:
$$
G(z, w)=-\frac{1}{2\pi}\log {\rm dist}(z, w) +R_{\rm robin}(w)+\mathcal{O}({\rm dist}(z, w)).
$$

As for the infinite kinetic energy, we conclude that it should be renormalized by replacing the diagonal
terms $G(w_k, w_k)$ by $R_{\rm robin}(w_k)$. This gives the same equations of motion as more direct approaches, 
or those available in the literature, like \cite{Hally-1980, Boatto-Koiller-2013, Dritschel-Boatto-2015}.
Thus (twice) the {\it renormalized energy} is
$$ 
2(\nu,\nu)_{1,{\rm renorm}}=\sum_{k=1}^n \Gamma_k^2R_{\rm robin}(w_k)+\sum_{k\ne j} \Gamma_k \Gamma_jG(w_k, w_j) +
$$
\begin{equation}\label{nunurenorm}
+\sum_{j=1}^{\texttt g}\big(\oint_{\alpha_j}\eta\oint_{\beta_j}*\eta-\oint_{\alpha_j}*\eta\oint_{\beta_j}\eta\big).
\end{equation}
This depends on the locations $w_1, \dots, w_n$ of the vortices (even the last term depends on these, although somewhat more indirectly). 

Many authors start out directly with the Robin function (\ref{Rrobin}), but there are some advantages with exposing the two terms in it
as individual quantities. One is that it clarifies the structure of the vortex motion equations by separating harmonic contributions, such as 
$h_0(w_k)$ and $G(w_k, w_j)$, from differential geometric contributions, like $\log \lambda (w_k)$. The first category of terms can be classified
as nonlocal, coming from solutions of elliptic partial differential equations on the entire surface, whereas the second category are purely local in nature. 
There is also the contribution from global circulations, represented by the final term in (\ref{nunurenorm}),
and this may be considered to be truly global. 
So the vortex motion is in principle governed by a balance between different categories of contributions, local, nonlocal,
and global in nature. As we shall see
later this balance is drastically changed in the more singular case of dipole motion (Sections~\ref{sec:vortex pairs} and \ref{sec:dipoles}). 
Then only the local terms survive.

We need to elaborate further the expression (\ref{nudG}) and relate it to the global circulations. 
The $1$-form $\eta$ is harmonic and hence can be expanded as
\begin{equation}\label{etaAB}
\eta=-\sum_{j=1}^\texttt{g} A_j dU_{\beta_j}+\sum_{j=1}^\texttt{g} B_j dU_{\alpha_j}.
\end{equation}
The coefficients $A_k$, $B_k$ are in view of  (\ref{taualpha}), (\ref{taubeta}) given by
$$
A_k=\oint_{\alpha_k}\eta, \quad
B_k=\oint_{\beta_k}\eta.
$$
The circulations of the total flow $\nu=\eta-*dG^\omega$  then become
$$
a_k=\oint_{\alpha_k} \nu =A_k-\oint_{\alpha_k}*dG^\omega,
$$
$$
b_k=\oint_{\beta_k}\nu=B_k-\oint_{\beta_k}*dG^\omega.
$$

We shall treat the circulations $a_1,\dots,a_{\texttt{g}}$, $b_1,\dots,b_{\texttt{g}}$ as free (independent)
variables, along with the locations $w_1,\dots, w_n$ of the vortices. These variables will be the coordinates of the phase space, 
and together they determine $\eta$ and $\nu$.
The locations of the vortices are complex variables, in contrast to the circulations which are real.

If $\gamma$ is a closed oriented curve in $M$, fixed in time and avoiding the point vortices, then 
according to the Euler equation (\ref{eulerinvariant}), and in the notation used there, 
the circulation of $\nu$ around $\gamma$ changes with speed
$$
\frac{d}{dt}\oint_\gamma \nu=\oint_\gamma \frac{\partial \nu}{\partial t}=\oint_\gamma (d\phi-\mathcal{L}_{\bf v}\nu)
=-\oint_\gamma i({\bf v}) d\nu=
$$
$$
=-\oint_\gamma i({\bf v})\omega = \oint_\gamma  i({\bf v})\big(\frac{\Gamma}{V}{\rm vol}\big)
=\frac{\Gamma}{V} \oint_\gamma *\nu.
$$ 
Here we used also (\ref{cartan}), (\ref{Hodgei}), (\ref{ddGGamma}).
It follows in particular, on choosing  $\gamma=\alpha_k, \beta_k$, that
\begin{equation}\label{dadtGamma} 
\frac{da_k}{dt}=\frac{\Gamma}{V} \oint_{\alpha_k} *\nu=\frac{\Gamma}{V} \oint_{\alpha_k} *\eta,
\end{equation}
\begin{equation}\label{dbdtGamma}
\frac{db_k}{dt}=\frac{\Gamma}{V} \oint_{\beta_k} *\nu=\frac{\Gamma}{V} \oint_{\beta_k} *\eta.
\end{equation}


\subsection{Matrix formalism and the Hamiltonian}\label{sec:matrix formalism} 

The period matrix (written here in block form)
$$
\left( \begin{array}{cc}
                                            P & R  \\
                                              R^T & Q  \\
\end{array} \right)=
\left( \begin{array}{cc}
                                             (- \oint_{\beta_k}*dU_{\beta_j}) & ( \oint_{\beta_k}*dU_{\alpha_j})  \\
                                              ( \oint_{\alpha_k}*dU_{\beta_j}) & ( -\oint_{\alpha_k}*dU_{\alpha_j})  \\
\end{array} \right)=
$$
\begin{equation}\label{ointdU}
=
\left( \begin{array}{cc}
                                             (\int_M dU_{\beta_k}\wedge*dU_{\beta_j}) & (- \int_M dU_{\beta_k}\wedge*dU_{\alpha_j})  \\
                                              ( -\int_M dU_{\alpha_k}\wedge*dU_{\beta_j}) & (  \int_M dU_{\alpha_k}\wedge*dU_{\alpha_j})  \\
\end{array} \right)
\end{equation}
is symmetric and positive definite (see \cite{Farkas-Kra-1992} in general). In particular, $P$ and $Q$ are themselves symmetric and positive definite.
As for $R$ and $R^T$ we need to be explicit with what are row and column indices above: in all entries above, 
$k$ is the row index, $j$ the column index. Thus, for example, $R_{kj}=\oint_{\beta_k}*dU_{\alpha_j}$. 

Next we write
$$
dU_\alpha=
\left(\begin{array}{c}
dU_{\alpha_1}\\
\vdots\\
dU_{\alpha_\texttt{g}}
\end{array}\right),
\quad
dU_\beta=
\left(\begin{array}{c}
dU_{\beta_1}\\
\vdots\\
dU_{\beta_\texttt{g}}
\end{array}\right).
$$
As mentioned, these two column matrices together define a basis of the harmonic forms.
Another basis is provided by the corresponding Hodge starred column vectors $*dU_\alpha$, $*dU_\beta$, 
in similar matrix notation. The relation between the bases is 
\begin{lemma}\label{lem:PQR}
The two bases $\{  dU_{\alpha},  dU_{\beta}\}$ and $\{  *dU_{\alpha},   *dU_{\beta}\}$
are related by 
\begin{equation}\label{RQPR} 
\left( \begin{array}{cc}
                                           *dU_{\alpha}  \\
                                            *dU_{\beta} \\
\end{array} \right)=
\left( \begin{array}{cc}
                                            R^T & Q  \\
                                              -P & -R \\
\end{array} \right)
\left( \begin{array}{cc}
                                              dU_{\alpha}  \\
                                              dU_{\beta} \\                                            
\end{array} \right).
\end{equation}
\end{lemma}

\begin{proof}
One simply checks that the two members in (\ref{RQPR}) have the same periods with respect to the homology basis  $\{\alpha_j, \beta_j\}$.
\end{proof}

We arrange also the circulations of the flow $\nu$ into column vectors:
$$
a=
\left(\begin{array}{c}
a_1\\
\vdots\\
a_\texttt{g}
\end{array}\right)=
\left(\begin{array}{c}
\oint_{\alpha_1}\nu\\
\vdots\\
\oint_{\alpha_\texttt{g}} \nu
\end{array}\right),
\quad
b=
\left(\begin{array}{c}
b_1\\
\vdots\\
b_\texttt{g}
\end{array}\right)=
\left(\begin{array}{c}
\oint_{\beta_1}\nu\\
\vdots\\
\oint_{\beta_\texttt{g}} \nu
\end{array}\right),
$$
briefly written as
\begin{equation}\label{anubnu}
a=\oint_\alpha \nu, \quad b=\oint_\beta \nu.
\end{equation}
Similarly for other vectors of circulations, for example
\begin{equation}\label{Aa}
A= \oint_\alpha \eta = \oint_\alpha \nu+ \oint_\alpha *dG^\omega=a + \oint_\alpha *dG^\omega,
\end{equation}
\begin{equation}\label{Bb}
B= \oint_\beta \eta = \oint_\beta \nu+ \oint_\beta *dG^\omega=b + \oint_\beta *dG^\omega.
\end{equation}

The $\alpha$-periods of the conjugated Green function were computed in (\ref{intalphajdG}), and from 
(\ref{ddwkint}), (\ref{ddwkintalphaj}) it follows that
\begin{equation*}\label{dAdw}
\frac{\partial A}{\partial w_k}dw_k+\frac{\partial A}{\partial \bar{w}_k}d\bar{w}_k= {\Gamma_k}\,dU_{\alpha}(w_k).
\end{equation*} 
Taking into account also the dependence of $a_1,\dots,a_\texttt{g}$, and doing the same for $B$ and the $\beta$-periods,
gives the total differentials
\begin{equation}\label{dAda}
dA=da+d\oint_\alpha *dG^\omega=da+\sum_{k=1}^n \Gamma_k dU_\alpha (w_k),
\end{equation}
\begin{equation}\label{dBdb}
dB=db+d\oint_\beta *dG^\omega=db+\sum_{k=1}^n \Gamma_k dU_\beta (w_k).
\end{equation}

In matrix notation (\ref{etaAB}) becomes
\begin{equation}\label{AdUBdU}
\eta=-A^T dU_\beta+ B^T dU_\alpha=      
\left(
 \begin{array}{cc}
B^T,&-A^T
\end{array}\right)
              \left(\begin{array}{c}
dU_\alpha\\
dU_\beta
\end{array}\right),
\end{equation}
and the contribution from $\eta$ to the kinetic energy is the quadratic form
$$
\int_M \eta \wedge *\eta= \sum_{k=1}^{\texttt g}(\oint_{\alpha_k}\eta\oint_{\beta_k}*\eta-\oint_{\beta_k}\eta\oint_{\alpha_k}*\eta)=
$$
$$ 
=
\left(
 \begin{array}{cc}
A^T,&B^T
\end{array}\right)
\left( \begin{array}{cc}
                                            P & R  \\
                                              R^T & Q  \\
\end{array} \right)
\left(\begin{array}{c}
A\\
B
\end{array}\right).
$$
The last equality is based on straight-forward computations.
Note that $P$, $Q$, $R$ are fixed matrices, while $A$ and $B$ depend on $w_1,\dots,w_n$,
$a_1,\dots,a_\texttt{g}$, $b_1,\dots, b_\texttt{g}$ via (\ref{Aa}),  (\ref{Bb}).

We now define the {\it Hamiltonian function}, $\mathcal{H}$, as the renormalized kinetic energy of the flow considered as a
function of the locations of the point vortices and of the circulations:
\begin{equation}\label{H}
2\mathcal{H}(w_1,\dots,w_n;a_1,\dots,a_{\texttt{g}},b_1,\dots,b_{\texttt{g}})=2(\nu,\nu)_{1,{\rm renorm}}=
\end{equation}
$$ 
\left( \begin{array}{cccc}
\Gamma_1,&\Gamma_2& \ldots &\Gamma_n
\end{array}\right)
\left( \begin{array}{cccc}
R_{\rm robin}(w_1)& G(w_1,w_2)&\ldots &G(w_1,w_n)  \\
                                              G(w_2,w_1) & R_{\rm robin}(w_2)&\ldots &G(w_2,w_n) \\
\vdots& \vdots & \ddots &\vdots\\
G(w_n, w_1)& G(w_n, w_2)&\dots &R_{\rm robin}(w_n)
\end{array} \right)
\left(\begin{array}{c}
\Gamma_1\\
\Gamma_2\\
\vdots\\
\Gamma_n
\end{array}\right)+
$$
$$
+\left(
 \begin{array}{cc}
A^T,&B^T
\end{array}\right)
\left( \begin{array}{cc}
                                            P & R  \\
                                              R^T & Q  \\
\end{array} \right)
\left(\begin{array}{c}
A\\
B
\end{array}\right).
$$
More formally one could consider the complex conjugates $\bar{w}_j$ of the vortex positions as independent variables and write
the Hamiltonian as 
$$
\mathcal{H}(w_1, \dots,w_n ,\bar{w}_1, \dots, \bar{w}_n ;a_1,\dots,a_{\texttt{g}},b_1,\dots,b_{\texttt{g}}).
$$ 


\section{Hamilton's equations}\label{sec:hamiltoneq} 

\subsection{Phase space and symplectic form}\label{sec:phase space}

To formulate Hamilton's equation one needs, besides the Hamiltonian function itself, a phase 
space and a symplectic form on it. The phase space will in our case consist of all possible configurations of the vortices,
collisions not allowed, together with all possible circulations of the flow around the curves in the homology basis.
Thus we take it to be
$$
\mathcal{F}=\{(w_1,\dots,w_n;a_1,\dots,a_{\texttt{g}}, b_1,\dots,b_\texttt{g}): w_j\in M, w_k\ne w_j \text{ for } k\ne j\}.
$$
Here the $w_j$ shall be interpreted as points on $M$, but in most equations below they will refer to complex coordinates
for such points. Compare Remark~\ref{rem:notation}.

Assuming that $\Gamma\ne 0$ (recall (\ref{totalGamma})) the {\it symplectic form} on $\mathcal{F}$ is taken to be 
$$
\Omega=\sum_{k=1}^n \Gamma_k{\rm vol}(w_k)-\frac{V}{\Gamma}\sum_{j=1}^\texttt{g}da_j\wedge db_j=
$$
\begin{equation}\label{symplectic form}
=-\frac{1}{2 \I}\sum_{k=1}^n \Gamma_k\lambda(w_k)^2 dw_k\wedge d\bar{w}_k
-\Gamma{V}\sum_{j=1}^\texttt{g}\frac{da_j}{\Gamma}\wedge \frac{db_j}{\Gamma}.
\end{equation}
The last expression makes sense also if $\Gamma=0$, because
by (\ref{dadtGamma}), (\ref{dbdtGamma}) the factors $da_j/\Gamma$ and $db_j/\Gamma$ remain finite as $\Gamma\to 0$.
Thus in both expressions above, the last term shall simply be removed if $\Gamma=0$.

Let  
$$
\xi=\sum_{k=1}^n(\dot{w}_k\frac{\partial}{\partial w_k}+\dot{\bar{w}}_k\frac{\partial}{\partial \bar{w}_k})
+\sum_{j=1}^\texttt{g}(\dot{a}_j\frac{\partial}{\partial a_j}+\dot{b}_j\frac{\partial}{\partial b_j})
$$
denote a generic tangent vector of $\mathcal{F}$ viewed as a derivation. 
As for the first term, recall the conventions in Remark~\ref{rem:real-complex}.
The {\it Hamilton equations} in general say that
\begin{equation}\label{ixiOmegadH}
i(\xi)\Omega=d\mathcal{H}.
\end{equation}
One main issue then is to verify that with our choices of phase space, symplectic form and Hamiltonian function,
the equations (\ref{ixiOmegadH}) really produce the expected vortex dynamics.  This will be 
accomplished in Section~\ref{sec:dynamics}. 

To evaluate (\ref{ixiOmegadH}) we first compute the left member as
\begin{equation}\label{left member}
i(\xi)\Omega= -\frac{1}{2 \I} \sum_{k=1}^n\Gamma_k\lambda(w_k)^2\big(\dot{w}_k d\bar{w}_k- \dot{\bar{w}}_k d{w}_k\big)
-\frac{V}{\Gamma}\sum_{j=1}^\texttt{g}(\dot{a}_jdb_j-\dot{b}_jda_j).
\end{equation}
Explicitly (\ref{ixiOmegadH}) therefore says that
\begin{equation}\label{hamiltonexplicit1}
\Gamma_k\lambda(w_k)^2{\dot{w}_k}=-2\I \frac{\partial\mathcal{H}}{\partial \bar{w}_k},
\end{equation}
\begin{equation}\label{hamiltonexplicit2}
\dot{a}_j=-\frac{\Gamma}{V}\frac{\partial\mathcal{H}}{\partial {b}_j},
\quad
\dot{b}_j=+\frac{\Gamma}{V}\frac{\partial\mathcal{H}}{\partial {a}_j}.
\end{equation}
Here the partial derivatives in the right members are implicit in 
\begin{equation}\label{HHHHH}
d\mathcal{H}=\sum_{k=1}^n\frac{\Gamma_k^2}{2} 
\big(\frac{\partial R_{\rm robin}(w_k)}{\partial w_k}dw_k+\frac{\partial R_{\rm robin}(w_k)}{\partial \bar{w}_k}d\bar{w}_k\big)+
\end{equation}
$$
+ \sum_{k,j=1,k\ne j}^n\frac{\Gamma_k\Gamma_j}{2} \big(\frac{\partial G(w_k,w_j)}{\partial {w}_k}d{w}_k
+\frac{\partial G(w_k,w_j)}{\partial \bar{w}_k}d\bar{w}_k\big)+
$$
$$
+ \sum_{k,j=1,k\ne j}^n\frac{\Gamma_k\Gamma_j}{2} \big(\frac{\partial G(w_k,w_j)}{\partial {w}_j}d{w}_j
+\frac{\partial G(w_k,w_j)}{\partial \bar{w}_j}d\bar{w}_j\big)+
$$
$$
+\left(
 \begin{array}{cc}
A^T,&B^T
\end{array}\right)
\left( \begin{array}{cc}
                                            P & R  \\
                                              R^T & Q  \\
\end{array} \right)
\left(\begin{array}{c}
\sum_{k=1}^n \Gamma_kdU_{\alpha}(w_k)\\ \\
\sum_{k=1}^n \Gamma_kdU_{\beta}(w_k)
\end{array}\right)+
$$ 
$$
+\left(
 \begin{array}{cc}
A^T,&B^T
\end{array}\right)
\left( \begin{array}{cc}
                                            P & R  \\
                                              R^T & Q  \\
\end{array} \right)
\left(\begin{array}{c}
da\\
db
\end{array}\right).
$$
Recall here the expressions (\ref{dAda}), (\ref{dBdb}) for $A$, $B$ in terms of the phase space variables
$a$, $b$ (as column matrices) and $w_1,\dots, w_n$, $\bar{w}_1,\dots, \bar{w}_n$.

In the partial derivatives of $R_{\rm robin}$ we single out the two {\it affine connections}
\begin{equation}\label{rr} 
r_{\rm metric}(w)=2\frac{\partial}{\partial w}\log\lambda(w),\qquad\qquad
\end{equation}
\begin{equation}\label{rrr}
r_{\rm robin}(w)\ =-2h_1(w)=-2\frac{\partial}{\partial w}h_0(w),
\end{equation}
(see (\ref{c1c0}) for the last equality). Recalling (\ref{Rrobin})
we then have the following alternative expressions for (ingredients of) the first term in $d\mathcal{H}$:
\begin{equation}\label{dRdw}
\frac{\partial R_{\rm robin}(w_k)}{\partial w_k}
=\frac{1}{4\pi}\big(r_{\rm metric}(w_k)-r_{\rm robin}(w_k)\big)=
\end{equation}
$$
=\frac{1}{2\pi} \big( {h_1(w_k)}+\frac{\partial}{\partial {w}_k}\log \lambda(w_k)\big).
$$
Similarly for the conjugated quantities.


\subsection{Dynamical equations}\label{sec:dynamics} 

The following theorem now makes the Hamilton equations (\ref{ixiOmegadH}) explicit in the vortex dynamics case.

\begin{theorem}\label{thm:dynamics}
The dynamical equations for the vortices $w_k$ and the circulations $a_j$, $b_j$ are, in matrix notation, 
\begin{align*}
\lambda(w_k)^2{\dot{w}_k}=&\frac{\Gamma_k}{2\pi \I} \big( \overline{h_1(w_k)}+\frac{\partial}{\partial \bar{w}_k}\log \lambda(w_k)\big)
-{2\I}\sum_{j=1,j\ne k}^n \Gamma_j \frac{\partial G(w_k,w_j)}{\partial \bar{w}_k}-\\
\\
&+2\left(
 \begin{array}{cc}
B^T,&-A^T
\end{array}\right)
              \left(\begin{array}{c}
{\partial U_\alpha(w_k)}/{\partial \bar{w}_k}\\
{\partial U_\beta(w_k)}/{\partial \bar{w}_k}
\end{array}\right),\\
\\
\left(\begin{array}{c}
\dot{a}\\
\dot{b}
\end{array}\right)=&
\frac{\Gamma}{V}
\left( \begin{array}{cc}
                                            -R^T & -Q \\
                                              P & R  \\
\end{array} \right)
\left(\begin{array}{c}
A\\
B
\end{array}\right).
\end{align*}
Recall the column vectors (\ref{Aa}), (\ref{Bb}):
\begin{equation*}
A=a + \oint_\alpha *dG^\omega,
\quad
B=b + \oint_\beta *dG^\omega.
\end{equation*}
\end{theorem}

\begin{proof}
The equations are obtained by identifying the expression  (\ref{left member}) for $i(\xi)\Omega$
with the expression (\ref{HHHHH}) for $d\mathcal{H}$.  
Thus, in (\ref{left member}), the $k$:th term in the first sum, 
\begin{equation}\label{k:th term} 
 -\frac{1}{2 \I} \Gamma_k\lambda(w_k)^2\big(\dot{w}_k d\bar{w}_k- \dot{\bar{w}}_k d{w}_k\big),
\end{equation}
is to be identified with the corresponding parts in the right member of (\ref{HHHHH}), namely the first three terms. 
Together with (\ref{dRdw}) this gives immediately the first two terms in the equation for $\dot{w}_k$. 

The third term comes from the term
$$
\left(
 \begin{array}{cc}
A^T,&B^T
\end{array}\right)
\left( \begin{array}{cc}
                                            P & R  \\
                                              R^T & Q  \\
\end{array} \right)
\left(\begin{array}{c}
\sum_{k=1}^n \Gamma_kdU_{\alpha}(w_k)\\ \\
\sum_{k=1}^n \Gamma_kdU_{\beta}(w_k)
\end{array}\right)
$$ 
in equation (\ref{HHHHH}). On using (\ref{RQPR}) this can be rewritten as
$$
\left(
 \begin{array}{cc}
B^T,&-A^T
\end{array}\right)
\left( \begin{array}{cc}
                                            R^T & Q  \\
                                              -P & -R  \\
\end{array} \right)
\left(\begin{array}{c}
\sum_{k=1}^n \Gamma_kdU_{\alpha}(w_k)\\ \\
\sum_{k=1}^n \Gamma_kdU_{\beta}(w_k)
\end{array}\right)=
$$ 
$$
=\left(
 \begin{array}{cc}
B^T,&-A^T
\end{array}\right)
\left(\begin{array}{c}
\sum_{k=1}^n \Gamma_k*dU_{\alpha}(w_k)\\ \\
\sum_{k=1}^n \Gamma_k*dU_{\beta}(w_k)
\end{array}\right).
$$ 
Identifying here the coefficient for $d\bar{w}_k$ with the corresponding coefficient in (\ref{k:th term}) gives 
$$   
\lambda(w_k)^2\dot{w}_k=2\left(
 \begin{array}{cc}
B^T,&-A^T
\end{array}\right)
              \left(\begin{array}{c}
\frac{\partial U_\alpha(w_k)}{\partial \bar{w}_k}\\
\frac{\partial U_\beta(w_k)}{\partial \bar{w}_k}\end{array}\right),
$$
as desired.

Finally, the term
$$
\left(
 \begin{array}{cc}
A^T,&B^T
\end{array}\right)
\left( \begin{array}{cc}
                                            P & R  \\
                                              R^T & Q  \\
\end{array} \right)
\left(\begin{array}{c}
da\\
db
\end{array}\right).
$$
in (\ref{HHHHH}) is 
to be identified with 
$$
-\frac{V}{\Gamma}
\left( \begin{array}{cc}
                                            -\dot{b}, & \dot{a}  \\
                                              
\end{array} \right)
\left( \begin{array}{cc}
                                            da  \\
                                              db \\
\end{array} \right)
$$
in (\ref{left member}). This gives
$$
\left(\begin{array}{c}
\dot{a}\\
\dot{b}
\end{array}\right)=
\frac{\Gamma}{V}
\left( \begin{array}{cc}
                                            -R^T & -Q \\
                                              P & R  \\
\end{array} \right)
\left(\begin{array}{c}
A\\
B
\end{array}\right),
$$
as desired.
\end{proof}

Besides the formal proof of the theorem it is essential to show that the dynamics given in Theorem~\ref{thm:dynamics}
really is the ``expected dynamics''.  For the vortices this expected dynamics is that each individual vortex, say $w_k$, 
moves with the speed of total flow, namely the $1$-form $\nu=\eta-*dG^\omega$ (see (\ref{nudG})) converted into a vector,
however with the modification that the contribution to $\nu$
from $w_k$ itself shall be regularized according to standard procedures involving the Robin function. 

In the notation of Remark~\ref{rem:real-complex}, the velocity ${\bf V}$ of the $k$:th vortex $w_k$ corresponds,
with $z=w_k$, to
the covector $\frac{\lambda^2}{2}(\dot{\bar{z}}dz+\dot{{z}}d\bar{z})$, see equations (\ref{vectorV}), (\ref{covectorV}).
Thus in the notation of the Theorem~\ref{thm:dynamics} the expected dynamics is expressed by the equation
$$
\eta-\{*dG^\omega\}_{\rm renormalized}=\frac{1}{2}\lambda(w_k)^2\big(\dot{\bar{w}}_kdw_k+{\dot{w}_k}d\bar{w}_k\big),
$$
the left member to be evaluated at $w_k$.
Here the second term in the left member is the well-known \cite{Llewellyn-2011, Boatto-Koiller-2013}
contribution from the Green function, where (referring to (\ref{dRdw}), (\ref{GGammaG}))
$$
\{*dG^\omega\}_{\rm renormalized}(z)=\Gamma_k*dR_{\rm robin}(z)+\sum_{j=1, j\ne k}^n\Gamma_j*dG(z,w_j),
$$
evaluated at $z=w_k$ and the Hodge star taken with respect to $z$. This gives
$$
-\{*dG^\omega\}_{\rm renormalized}(w_k)
=\I\Gamma_k\Big(\frac{\partial R_{\rm robin}(w_k)}{\partial w_k}dw_k-\frac{\partial R_{\rm robin}(w_k)}{\partial \bar{w}_k}d\bar{w}_k\Big)+
$$
$$
+{\I}\sum_{j=1,j\ne k}^n \Gamma_j\Big(\frac{\partial G(w_k,w_j)}{\partial {w}_k}dw_k- \frac{\partial G(w_k,w_j)}{\partial \bar{w}_k}d\bar{w}_k\Big)
$$

The above accounts for the first two terms in the right member of Theorem~\ref{thm:dynamics}. 
The third term is the contribution from $\eta$, and can by (\ref{AdUBdU}) be directly identified as $2\eta_{\bar{z}}$ evaluated at $z=w_k$.

Finally we verify that the periods change according to (\ref{dadtGamma}), (\ref{dbdtGamma}).
By (\ref{AdUBdU}), (\ref{RQPR}) we have 
$$
*\eta=
\left( \begin{array}{cc}
                                            B^T, & -A^T  \\
                                              
\end{array} \right)
\left( \begin{array}{cc}
                                            *dU_{\alpha}  \\
                                              *dU_{\beta} \\
\end{array} \right)
=\left( \begin{array}{cc}
                                            B^T, & -A^T  \\
                                              
\end{array} \right)
\left( \begin{array}{cc}
                                            R^T & Q  \\
                                              -P & -R \\
\end{array} \right)
\left( \begin{array}{cc}
                                              dU_{\alpha}  \\
                                              dU_{\beta} \\                                            
\end{array} \right)=
$$
$$
=
\left( \begin{array}{cc}
                                            (dU_\alpha)^T, & (dU_\beta)^T \\                                          
\end{array} \right)
\left( \begin{array}{cc}
                                            P&R\\
                                            R^T & Q  \\
                                            
\end{array} \right)
\left( \begin{array}{cc}
                                              A \\
                                             B \\                                            
\end{array} \right).
$$
By (\ref{taualpha}), (\ref{taubeta}) integration of this gives
$$
\left( \begin{array}{cc}
                                              \oint_\alpha*\eta \\
                                              \oint_\beta*\eta \\                                            
\end{array} \right)
=
\left( \begin{array}{cc}
                                           
                                           - R^T & -Q  \\
                                            P & R
\end{array} \right)
\left( \begin{array}{cc}
                                              A \\
                                              B \\                                            
\end{array} \right).
$$
From this we see that the last equation in the theorem is in exact agreement with (\ref{dadtGamma}), (\ref{dbdtGamma}),
as claimed.


\section{Motion of a single point vortex}\label{sec:single vortex}

In the case of a single vortex Theorem~\ref{thm:dynamics} simplifies a little. We may then denote the vortex point $w_1$
simply by $w$, and the strength $\Gamma_1$ agrees with the total vorticity $\Gamma$ for the point vortices. 
If in addition $\texttt{g}=0$ then everything simplifies considerable. There is only one free variable, the location $w\in M$
for the vortex, and the dynamical equation for this is
$$
\lambda(w)^2\frac{dw}{dt}=\frac{\Gamma}{2\pi \I} \big( \overline{h_1(w)}+\frac{\partial}{\partial \bar{w}}\log \lambda(w)\big).
$$
The Hamiltonian and the symplectic form are
$$
\mathcal{H}(w)=\Gamma^2 R_{\rm robin}(w)=\frac{\Gamma^2}{2\pi}(h_0(w)+\log \lambda(w)),
$$
$$
\Omega=\Gamma \,{\rm vol}(w)= -\frac{1}{2\I}\Gamma \lambda(w)^2 dw\wedge d\bar{w}.
$$
It follows that if (and only if) the two metrics
\begin{equation}\label{dsmetric}
ds_{\rm metric}^2 =\lambda(w)^2|dw|^2=\frac{2V}{\pi}h_{11}(w)|dw|^2,
\end{equation}
\begin{equation}\label{dsrobin}
ds_{\rm robin}^2\ =e^{-2h_0(w)}|dw|^2 \qquad\qquad\qquad\quad
\end{equation}
are identical, up to a constant factor, then the vortex will never move, whatever its initial position is. 
In \cite{Boatto-Koiller-2013, Grotta-Ragazzo-Barros-Viglioni-2017} this is referred to as 
$ds_{\rm metric}$ being a ``steady vortex metric'', or being ``hydrodynamically neutral''.  

\begin{example}\label{ex:sphere}
An obvious example is a homogenous sphere. Indeed, for a sphere of radius one 
we have, in coordinates obtained by stereographic projection into the complex plane, 
well-known formulas such as
$$
G(z,w)=-\frac{1}{4\pi}\Big(\log \frac{|z-w|^2}{(1+|z|^2)(1+|w|^2)}+1\Big),
$$
$$
h_0(w)=\log (1+|w|^2)-\frac{1}{2},\quad
h_1(w)=\frac{\bar{w}}{1+|w|^2},\quad
h_2(w)=\frac{\bar{w}^2}{2(1+|w|^2)^2},
$$
$$
e^{-2h_0(w)}=\frac{e}{(1+|w|^2)^2},\quad
\lambda(w)^2=8h_{11}(w)= \frac{4}{(1+|w|^2)^2}. 
$$
The Hamiltonian function is constant, 
$$
\mathcal{H}(w)=\frac{\Gamma^2}{2\pi}\big(\log (1+|w|^2)-\frac{1}{2}+\log \frac{2}{1+|w|^2}\big)=\frac{\Gamma^2}{4\pi}(2\log 2-1),
$$
and there is no motion of the vortex.
\end{example} 

When $\texttt{g}>0$ one has to include also the circulating flows in the picture, so that the dynamical equations become, by
Theorem~\ref{thm:dynamics},
\begin{align*}
\lambda(w)^2\dot{w}=&\frac{\Gamma}{2\pi \I} \big( \overline{h_1(w)}+\frac{\partial}{\partial \bar{w}}\log \lambda(w)\big)
+2\left(
 \begin{array}{cc}
B^T,&-A^T
\end{array}\right)
              \left(\begin{array}{c}
{\partial U_\alpha(w)}/{\partial \bar{w}}\\
{\partial U_\beta(w)}/{\partial \bar{w}}
\end{array}\right),\\
\\
\left(\begin{array}{c}
\dot{a}\\
\dot{b}
\end{array}\right)=&
\frac{\Gamma}{V}
\left( \begin{array}{cc}
                                            -R^T & -Q \\
                                              P & R  \\
\end{array} \right)
\left(\begin{array}{c}
A\\
B
\end{array}\right).
\end{align*}


\section{Motion of a vortex pair in the dipole limit}\label{sec:vortex pairs} 

For a vortex pair $\{w_1, w_2\}$ with $\Gamma_1=-\Gamma_2$ we have $\Gamma=\Gamma_1+\Gamma_2=0$, 
hence there is no compensating background vorticity. The circulations $a$ and $b$ will be time independent 
by the last equation in Theorem~\ref{thm:dynamics} and are not needed in phase
space, which then simply becomes
$$
\mathcal{F}=\{(w_1,w_2): w_1,w_2 \in M, w_1\ne w_2\},
$$ 
with symplectic form
\begin{equation}\label{symplectic form1}
\Omega= -\frac{1}{2\I}\Gamma_1\big( \lambda(w_1)^2d w_1\wedge d\bar{w}_1 -\lambda(w_2)^2 dw_2\wedge d\bar{w}_2\big).
\end{equation}
The Hamiltonian is the same quantity as before, see (\ref{H}), but it may now be considered as a function only of $w_1$ and $w_2$. 
The circulations $a$, $b$ are fixed parameters, given in advance.

However the period vectors $A$ and $B$ still depend on time
via $w_1$, $w_2$. Indeed, using (\ref{ointdVa}), (\ref{ointdVb}) we have
$$
A=a+\Gamma_1\oint_{\alpha}\big(*dG(\cdot,w_1)-*dG(\cdot,w_2)\big)=a+\Gamma_1\int_{w_2}^{w_1}dU_{\alpha} ,
$$
$$
B=b+\Gamma_1\oint_{\beta}\big(*dG(\cdot,w_1)-*dG(\cdot,w_2)\big)=b+\Gamma_1\int_{w_2}^{w_1}dU_{\beta_j} . 
$$
The Hamiltonian is 
$$
2\mathcal{H}(w_1,w_2)=
\left(
 \begin{array}{cc}
\Gamma_1,&-\Gamma_1
\end{array}\right)
\left( \begin{array}{cc}
                                            R_{\rm robin}(w_1) & G(w_1,w_2)  \\
                                              G(w_2,w_1) &R_{\rm robin}(w_2)\\
\end{array} \right)
\left(\begin{array}{c}
\Gamma_1\\
-\Gamma_1
\end{array}\right)+
$$
\begin{equation}\label{hamiltonian}
+\left(
 \begin{array}{cc}
A^T,&B^T
\end{array}\right)
\left( \begin{array}{cc}
                                            P & R  \\
                                              R^T & Q  \\
\end{array} \right)
\left(\begin{array}{c}
A\\
B
\end{array}\right),
\end{equation}
and the dynamics of the vortex pair $\{w_1,w_2\}$ becomes, by (\ref{hamiltonexplicit1}), 
$$
\Gamma_1 \lambda(w_1)^2 \,\dot{w}_1= -2\I \,\frac{\partial\mathcal{H}(w_1,w_2)}{\partial\bar{w}_1},
$$
$$
\Gamma_1 \lambda(w_2)^2 \,\dot{w}_2= +2\I \,\frac{\partial\mathcal{H}(w_1,w_2)}{\partial\bar{w}_2}.
$$

In place of $w_1$ and $w_2$ we may turn to $w=\frac{1}{2}(w_1+w_2)$ and $u=\frac{1}{2}(w_1-w_2)$ as coordinates.
These are similar to the ``center-arrow coordinates'' used in \cite{Boatto-Koiller-2013, Llewellyn-Nagem-2013}. Then 
\begin{equation}\label{wpmu}
\begin{cases}
w_1=w+u,\\
w_2=w-u.
\end{cases}
\end{equation}
We are interested in the limit $u\to 0$, and
Taylor expansion of $H(w_1,w_2)=H(w+u, w-u)$ with respect to $u$, $\bar{u}$ gives, using relations in Appendix~2,
$$
H(w+u,w-u)=h_0(w)-\re\big((4h_2(w)-\frac{\partial h_1(w)}{\partial w})u^2\big)+
$$
\begin{equation}\label{expansionH}
+(4h_{11}(w)-\frac{\partial h_1(w)}{\partial\bar{w}})|u|^2 +\mathcal{O}(|u|^3).
\end{equation}
The linear terms vanish because of the symmetry of $H(w_1,w_2)$. 
The second order terms will only be used briefly in Section~\ref{sec:dipoles}, and 
even the constant term $h_0(w)$ will eventually disappear below.
In addition to (\ref{expansionH}) we have the Taylor expansions
$$
h_0(w\pm u)=h_0(w)\pm \big(h_1(w)u+\overline{{h}_1(w)u}\big)+\mathcal{O}(|u|^2),
$$
$$
\log \lambda(w \pm u)=\log\lambda(w) \pm \frac{1}{2}\big(r(w)u+\overline{r(w)u}\big)+\mathcal{O}(|u|^2)
$$
(coupled signs throughout). The latter equation uses the affine connection $r(w)=r_{\rm metric}(w)$, see 
(\ref{rr}) or (\ref{rlog}). For later use we record also the expansion
\begin{equation}\label{expansion lambda}
\lambda(w\pm u)^2=\lambda(w)^2\Big(1\pm \big(r(w)u +\overline{r(w)}\bar{u}\big)+\mathcal{O}(|u|^2)\Big). 
\end{equation}

Using these expansions we obtain, for the first matrix in the Hamiltonian (\ref{hamiltonian}),
\begin{equation}\label{hamiltonianmatrix}
2\pi \left( \begin{array}{cc}
                                            R_{\rm robin}(w_1) & G(w_1,w_2)  \\
                                              G(w_2,w_1) &R_{\rm robin}(w_2)\\
\end{array} \right)=
\end{equation}
$$
=\left( \begin{array}{cc}
                                            h_0(w+u)+\log\lambda(w+u) & -\log |2u| +H(w+u,w-u)\\
                                               -\log | 2u| +H(w+u,w-u) &h_0(w-u)+\log\lambda(w-u)\\
\end{array} \right)=
$$
$$
=\left( \begin{array}{cc}
                                            \log\lambda(w) & -\log |2u|\\
                                               -\log | 2u|  &\log\lambda(w)\\
\end{array} \right)
+\left( \begin{array}{cc}
                                            h_0(w) & h_0(w)\\
                                               h_0(w) &h_0(w)\\
\end{array} \right)+
$$
$$
+\left( \begin{array}{cc}
                                           h_1(w)+\frac{1}{2}r(w) & 0\\
                                               0 &-h_1(w)-\frac{1}{2}r(w)\\
\end{array} \right)u
+\left( \begin{array}{cc}
                                           \overline{{h}_1(w)}+\frac{1}{2}\overline{r(w)} & 0\\
                                               0 &-\overline{{h}_1(w)}-\frac{1}{2}\overline{r(w)}\\
\end{array} \right)\bar{u}+
$$
$$
+\mathcal{O}(|u|^2).
$$

When acting with $(\Gamma_1, -\Gamma_1)$ on both sides of the matrix (\ref{hamiltonianmatrix}), the last three terms 
in the final expression disappear and the result becomes, up to $\mathcal{O}(|u|^2)$,
\begin{equation}\label{hamiltonianmatrix1}
\left( \begin{array}{cc}
                                            \Gamma_1 & -\Gamma_1
                                      
\end{array} \right)
\left( \begin{array}{cc}                                            \log\lambda(w) & -\log |2u| \\
                                               -\log | 2u| &\log\lambda(w)\\
\end{array} \right)
\left( \begin{array}{cc}                                            \Gamma_1 \\
                                               -\Gamma_1 \\
\end{array} \right)=
\end{equation}
$$
=2\Gamma_1^2\big(\log\lambda(w)+\log |2u|\big).
$$

The full Hamiltonian (\ref{hamiltonian}) therefore becomes, up to terms of order $\mathcal{O}(|u|^2)$,
$$
2\mathcal{H}(w+u,w-u)=\frac{\Gamma_1^2}{\pi}\big(\log\lambda(w)+\log |2u|\big)+
$$
$$
+\left(
 \begin{array}{cc}
(a+\Gamma_1 \int_{w-u}^{w+u}dU_\alpha)^T,&(b+\Gamma_1 \int_{w-u}^{w+u}dU_\beta)^T\end{array}\right)
\left( \begin{array}{cc}
                                            P & R  \\
                                              R^T & Q  \\
\end{array} \right)
\left(\begin{array}{c}
a+\Gamma_1 \int_{w-u}^{w+u}dU_\alpha\\
b+\Gamma_1 \int_{w-u}^{w+u}dU_\beta
\end{array}\right).
$$
Here one can see that the final term remains bounded as $u\to 0$, hence is  negligible in this limit
compared to the first term. Therefore the  leading terms  in this limit are given by 
\begin{equation}\label{Hwpmu}
\mathcal{H}(w+u,w-u)=\frac{\Gamma_1^2}{2\pi}\big(\log\lambda(w)+\log |2u|\big)+\mathcal{O}(1) \quad (u\to 0).
\end{equation}
This is essentially a constant factor times log of the distance (in the metric) between $w_1=w+u$ and $w_2=w-u$.
Indeed, we recover the simple and beautiful formula 
$$
\mathcal{H}(w_1,w_2)=\frac{\Gamma_1^2}{2\pi}\log {\rm dist\,}(w_1,w_2)+\mathcal{O}(1) \quad (|w_1-w_2|\to 0)
$$
of Boatto and Koiller. See equations (24), (25) in \cite{Boatto-Koiller-2013}. Compare also \cite{Grotta-Ragazzo-Barros-Viglioni-2017}.
The distance is taken with respect to the Riemannian metric. The error term
$\mathcal{O}(1)$ can be identified with a what is called the ``Batman function'' in \cite{Boatto-Koiller-2013}. 
In our notations the latter is given by (\ref{qpolarized}).

Taking the differential of (\ref{Hwpmu}) gives, on using again
the metric affine connection $r(w)=r_{\rm metric}(w)$ defined by (\ref{rr}),
$$
d\mathcal{H}=\frac{\partial \mathcal{H}}{\partial w}dw+\frac{\partial \mathcal{H}}{\partial \bar{w}}d\bar{w}
+\frac{\partial \mathcal{H}}{\partial {u}}du+\frac{\partial \mathcal{H}}{\partial \bar{u}}d\bar{u}=
$$
\begin{equation}\label{rdwrdw}
=\frac{\Gamma_1^2}{4\pi}\Big( r(w)dw+\overline{r(w)}d\bar{w}
+\frac{du}{ u}+\frac{d\bar{u}}{\bar{u}}\Big). 
\end{equation}
Recalling (\ref{expansion lambda}) we can expand the symplectic $2$-form given by (\ref{symplectic form1}) in terms of $w$ and $u$ as
$$
\Omega=-\frac{\Gamma_1}{2\I}\big(\lambda(w+u)^2d(w+u)\wedge d(\bar{w}+\bar{u})-\lambda(w-u)^2d(w-u)\wedge d(\bar{w}-\bar{u})=
$$
$$
=\I{\Gamma_1}\lambda(w)^2\Big(dw\wedge d\bar{u}- d\bar{w}\wedge du  
+ \big(r(w)u+\overline{r(w)u}\big)
\big(dw\wedge d\bar{w}+du\wedge d\bar{u}\big)\Big)+\mathcal{O}(|u|^2).
$$
With
$$
\xi=\dot{w}\frac{\partial}{\partial w}+\dot{\bar{w}}\frac{\partial}{\partial \bar{w}}
+\dot{u}\frac{\partial}{\partial u}+\dot{\bar{u}}\frac{\partial}{\partial \bar{u}}
$$
this gives, as $u\to 0$,
$$
i(\xi)\Omega=
\I \Gamma_1 \lambda(w)^2 
\Big(
-\big(\dot{\bar{u}} +(r(w)u+\overline{r(w)u})\dot{\bar{w}}\big)dw
+\big( \dot{u}+(r(w)u+\overline{r(w)u})\dot{w}\big)d\bar{w}-
$$
$$
-\big(\dot{\bar{w}}+(r(w)u+\overline{r(w)u})\dot{\bar{u}}\big)du
+\big( \dot{w}+(r(w)u+\overline{r(w)u}\big)\dot{u})d\bar{u}+\mathcal{O}(|u|^2\Big).
$$
Comparing with (\ref{rdwrdw}) we see that the dynamics of the vortex pair is described by the two equations 
\begin{equation}\label{dotu}
\frac{\Gamma_1}{4\pi\I}\overline{r(w)}=\lambda(w)^2\big(\dot{{u}} +(r(w)u+\overline{r(w)u})\dot{{w}}\big)+\mathcal{O}(|u|^2),
\end{equation}
\begin{equation}\label{dotw}
\frac{\Gamma_1}{4\pi\I} \frac{1}{\bar{u}}=\lambda(w)^2\big(\dot{{w}}+(r(w)u+\overline{r(w)u})\dot{{u}}\big)+\mathcal{O}(|u|^2).
\end{equation}

Equation (\ref{dotw}) can used to eliminate $\dot{w}$ in (\ref{dotu}), which then becomes 
$$
\frac{\Gamma_1}{4\pi\I}\overline{r(w)}=\lambda(w)^2\dot{u}+\big(r(w)u+\overline{r(w)u}\big)
\Big(\frac{\Gamma_1}{4\pi\I\bar{u}}-\lambda(w)^2\big(r(w)u+\overline{r(w)u}\big)\dot{u} \Big)
+\mathcal{O}(|u|^2).
$$
Here the left member cancels with one of the terms in the right member, and the rest can be written, after division by $u$, 
$$
0=\lambda(w)^2\cdot\frac{\dot{u}}{u}+r(w)\frac{\Gamma_1}{4\pi\I\bar{u}}-\lambda(w)^2\big(r(w)u+\overline{r(w)u}\big)^2\cdot\frac{\dot{u}}{u} 
+\mathcal{O}(|u|).
$$
In this equation the third term in the right member is of a smaller magnitude than the other two terms
and can be incorporated in the final $\mathcal{O}(|u|)$.
Thus we arrive at 
\begin{equation}\label{dotlogu}
\lambda(w)^2\frac{d}{dt}\log u+r(w)\frac{\Gamma_1}{4\pi\I\bar{u}}=0,
\end{equation}
valid with an error of at most $\mathcal{O}(|u|)$ as $u\to 0$.

The above equation, (\ref{dotlogu}), essentially comes from (\ref{dotu}), and it is to be combined again with
(\ref{dotw}). For the latter it is enough to use the simplified form
\begin{equation}\label{dotu1}
\frac{\Gamma_1}{4\pi\I\bar{u}} =\lambda(w)^2\frac{{dw}}{dt},
\end{equation}
which only uses the leading terms, and for which the error is still at most $\mathcal{O}(|u|)$.
Inserting (\ref{dotu1}) in (\ref{dotlogu}) results in the master equation
\begin{equation}\label{dlogudt}
\frac{d}{dt}\log u+r(w)\frac{dw}{dt}=0.
\end{equation}

One problem with (\ref{dlogudt}) is that the speed $dw/dt$ becomes infinite, along with the first term, in the dipole limit. 
However, this problem only affects the real part of the equation. 
For the imaginary part one can
replace true time $t$ by an arbitrary parameter, which is scaled with $u$ so that $dw/dt$ remains finite
as $|u|\to 0$. Alternatively one may scale $\Gamma_1$ with $u$ so that the left member of (\ref{dotu1})
remains finite. Then one can still think of any new parameter  $t$ as a time variable.
In any case, we take imaginary parts  of (\ref{dlogudt}) and obtain
\begin{equation}\label{dargudt}
\frac{d}{dt}\arg u+\im (r(w)\frac{d{w}}{dt})=0.
\end{equation} 

Equation (\ref{dotu1}) shows that the directions of $u$ and $dw/dt$ are related as
\begin{equation}\label{arguargdwdt}
\arg u= \arg \frac{dw}{dt}\pm \frac{\pi}{2},
\end{equation}
where the plus sign is to be chosen if $\Gamma_1>0$, the minus sign if $\Gamma_1<0$.
Now (\ref{dargudt}) and (\ref{arguargdwdt}) taken together give the final law
for the motion of the center $w$ of the vortex pair in the dipole limit:
\begin{equation}\label{ddtargdwdt0}
\frac{d}{dt}\arg \frac{dw}{dt}+\im (r(w)\frac{d{w}}{dt})=0.
\end{equation} 
As explained in Appendix~1, Section~\ref{sec:connections} (see in particular equation (\ref{imgeodesics}) there), (\ref{ddtargdwdt}) is exactly 
the equation for a geodesic curve when expressed in an arbitrary parameter $t$. 
One may notice that (\ref{dargudt}) (and similarly for (\ref{ddtargdwdt0})) can be written in the parameter-free form 
$$
{d}\arg u+\im (r(w){d{w}})=0,
$$
confirming again the fact that the geometry of the dipole trajectory has a meaning independent of any choice of parameter.

The real part of (\ref{dlogudt}) says, in view of (\ref{rlog}) (or (\ref{rr})), that
$$
\frac{d}{dt}\big(\log|u|+ \log\lambda(w)\big)=0,
$$
the error term $\mathcal{O}(|u|)$ being disregarded.
In other words that \,${|u|}{\lambda (w)}= C$ ({constant}) along each trajectory.
By (\ref{dotu1}) this also gives 
\begin{equation}\label{speed}
\Big|\frac{dw}{dt}\Big| = \frac{C}{\lambda(w)}.
\end{equation}
Thus along each trajectory $\lambda(w)$ has, being proportional to one over the velocity, the same role as the refraction index in optics.

We summarize the most essential parts of the above discussion as

\begin{theorem}\label{thm:kimura}
The dynamical equations for a vortex pair in the dipole limit is identical with the geodesic equation for the metric $ds=\lambda(w)|dw|$, namely
\begin{equation}\label{ddtargdwdt}
\frac{d}{dt}\arg \frac{dw}{dt}+\im (r(w)\frac{d{w}}{dt})=0.
\end{equation} 
Here $w=w(t)$ is the location of the dipole, $t$ is an arbitrarily scaled time parameter chosen such that $dw/dt$ is finite.
The orientation of the dipole is perpendicular to $dw/dt$.
\end{theorem}

\begin{remark}
It is possible to understand why dipole move along geodesics by thinking of vortex pair as a ``wave front'',
in an optical analogy. Equation (\ref{arguargdwdt}) says that the motion is perpendicular to the wave front
(the line segment from $w-u$ to $w+u$). Equation (\ref{dargudt}) then expresses that if the front of a vortex pair is
not aligned with the level line of $\lambda(w)$ then the direction of $u$ (representing the wave front) changes
in such a way that the curve $w(t)$ bends towards higher values of $\lambda$.

Being slightly more direct and exact,
on taking $t$ to be Euclidean arc length the first term in (\ref{ddtargdwdt}) is the ordinary curvature for the curve traced out by $w(t)$.
The second term can be viewed as the inner product between the gradient of $\lambda(w)$, which can be identified with
$\overline{r(w)}$, and $dw/dt$ rotated $90$ degrees to the  right. 
Letting $\theta$ denote the angle between the gradient of $\lambda(w)$ and the velocity vector $dw/dt$ we can write
$$
\im (r(w)\frac{d{w}}{dt})= |\overline{r(w)}  |  |\frac{dw}{dt}| \sin {\theta}.
$$ 
The above remarks are compatible with the laws of optics, for example Fermat's law, and also with ``Snell's law'' (see for example 
\cite{Guillemin-Sternberg-1984, Cawte+-2019}) in the somewhat singular case that $\lambda(w)$
jumps between two constant values.  
\end{remark}


\section{Vortex dipoles and projective connections}\label{sec:dipoles} 

A vortex dipole constructed as a vortex pair melting together is a result of two regularizations: First, a single vortex is already 
regularized in itself by removal of a singularity of type $\log|z-w|$ in the stream function, together with a possible adjustment
with a term related to the curvature of the surface.
This is not so severe, and everyone agrees on the result. The motion of the vortex then comes from the regular terms 
in the Taylor expansion, and it is easy to handle. In particular the speed of the vortex is finite.

When the two vortices in a vortex pair approach each, then a second  
regularization becomes necessary. If one does not want the two vortices to just annihilate each other
in the limit, then one has to let the strengths of the vortices tend to infinity. Even without that, the speed of the vortex
pair will be infinite in the limit, so with the vortex strengths blowing up the speed will become infinite to an even higher degree.
But in any case one can speak of a residual trajectory, which will be a geodesic for the metric of the manifold, as was conjectured in
\cite{Kimura-1999},  proved in \cite{Boatto-Koiller-2013}, and further clarified  in Theorem~\ref{thm:kimura} above.

There is also the possibility of treating the flow directly as a dipole flow, not going via vortex pairs, and this touches on the theory of
projective (or Schwarzian) connections. Below we shall sketch upon such a procedure, but it seems difficult to read off the
trajectory of the dipole from this approach. 

We start from the expression, using (\ref{AdUBdU}),
$$
\nu=\eta-*dG^\omega =
\left( \begin{array}{cc}
                                            B^T, & -A^T  \\
                                              
\end{array} \right)
\left( \begin{array}{cc}
                                            dU_{\alpha}  \\
                                            dU_{\beta} \\
\end{array} \right)
-\Gamma*dG(\cdot,w)
$$
$$
=\left( \begin{array}{cc}
                                            b^T+\oint_{\beta}*dG(\cdot,w), & -a^T-\oint_{\alpha}*dG(\cdot,w)\\
                                              
\end{array} \right)
\left( \begin{array}{cc}
                                            dU_{\alpha}  \\
                                              dU_{\beta} \\
\end{array} \right)
-\Gamma*dG(\cdot,w)
$$
for the flow $1$-form $\nu$ in the single vortex case, and we shall to pass to the dipole limit by differentiation with respect to $w$.
Recall that $\nu$ is a real-valued form: on writing $\nu=\nu_z dz+\nu_{\bar{z}}d\bar{z}$ we have
$\nu_{\bar{z}}=\overline{\nu_z}$. When differentiating with respect to $w$ we get a covariant tensor, denoted
$d_w\nu$, which is real valued in a similar sense. In $d_w\nu$ there will appear  two kinds of differentials:  
$dz$, $d\bar{z}$ are differentials for the flow itself (as a $1$-form) at a general point $z$. 
The dipole is located at $w$ and the differentials $dw$, $d\bar{w}$  represent its orientation. These two kinds of differentials
are combined only via a tensor product (not an antisymmetric wedge product).

On using (\ref{stardG0}), (\ref{dUstardU}), (\ref{dUstardUminus}) one finds that differentiation of $\nu$ with respect to $w$ gives,
writing for simplicity $dzdw$ in place of $dz\otimes dw$ (etc),
$$
d_w \nu(z)=\frac{1}{2}\sum_{j=1}^\texttt{g}\big(dU_{\beta_j}(w)\otimes dU_{\alpha_j}(z)-dU_{\alpha_j}(w)\otimes dU_{\beta_j}(z)\big)+
$$
$$
+\I\Gamma\big( \frac{\partial^2 G(z,w)}{\partial z\partial w}dz dw + \frac{\partial^2 G(z,w)}{\partial z\partial \bar{w}}dz d\bar{w}
- \frac{\partial^2 G(z,w)}{\partial \bar{z}\partial {w}}d\bar{z} d{w} - \frac{\partial^2 G(z,w)}{\partial \bar{z}\partial \bar{w}}d\bar{z} d\bar{w}\big).
$$
Here the first term vanishes in the limit $z\to w$. In the remaining terms we single out the poles as separate terms and let $z\to w$ in what remains. This gives,
on using  (\ref{Hh2}), (\ref{Hh11}), (\ref{Hlambda}) and representing the first term simply by $\mathcal{O}(|z-w|)$ (at the end),
$$
d_w \nu(z)=-\frac{\Gamma}{4\pi\I}\Big(\big(\frac{1}{(z-w)^2}+2 \frac{\partial^2 H(z,w)}{\partial z\partial w}\big)dz dw + 2\frac{\partial^2 H(z,w)}{\partial z\partial \bar{w}}dz d\bar{w}-
$$
$$
-2 \frac{\partial^2 H(z,w)}{\partial \bar{z}\partial {w}}d\bar{z} d{w} 
-\big(\frac{1}{(\bar{z}-\bar{w})^2}+2 \frac{\partial^2 H(z,w)}{\partial \bar{z}\partial \bar{w}}\big)d\bar{z} d\bar{w}\Big)+\mathcal{O}(|z-w|)=
$$
$$
=-\frac{\Gamma}{2\pi}\im\Big(\big(\frac{1}{(z-w)^2}+ \frac{\partial h_1(w)}{\partial w}-2h_2(w)\big)dz dw
-\big( \frac{\partial h_1(w)}{\partial \bar{w}}-2h_{11}(w)\big))dz d\bar{w}\Big)
$$
$$
+\mathcal{O}(|z-w|).
$$

To get the speed of dipole one would like to insert $z=w$ above, but clearly this does not give any sensible result, just a flow $1$-form
in $z$ that becomes infinite at $z=w$. Still this singular flow is associated with a certain direction determined by $dw$, or more precisely by
the action of $dw$ on a vector $\texttt{m}$ representing the orientation of the dipole. See Example~\ref{rem:dipole singularity} below. 
In any case, this singular term overrules all other terms.

Up to a real factor, we thus have above the singularity
$$
\im\frac{dz dw}{(z-w)^2}
$$
plus the finite quantity
\begin{equation}\label{finite part}
\im\Big(\big( \frac{\partial h_1(w)}{\partial w}-2h_2(w)\big)dz dw 
- \big( \frac{\partial h_1(w)}{\partial \bar{w}}-2h_{11}(w)\big) dz d\bar{w}\Big).
\end{equation}
Certainly this finite term will be of minor importance compared with the singular term. The singular term actually determines the
entire flow (at a given instance) once the orientation of the dipole is given and, for example, the imaginary parts of the periods are
given, say are set to zero. Indeed, on letting $dw$ act on a vector $\texttt{m}$ representing the orientation, $\langle dw,{\texttt{m}}\rangle$ 
becomes a complex number and there remains the differentials $dz$ and $d\bar{z}$ representing the flow, and this can then be 
associated with with a unique meromorphic differential on $M$. 

To elaborate the above statements a little, let $\texttt{m}\in\C$ be given. Then
there exists a unique meromorphic differential $\mu=\mu_z dz$ on $M$ of the form
$$
\mu=\frac{{\texttt m}\,dz}{(z-w)^2}+{\rm regular}
$$
and having periods
$$
\im \oint_{\alpha_k} \mu=0, \quad \im \oint_{\beta_k} \mu=0 \quad (k=1,\dots, \texttt{g}).
$$
Such a $\mu$ represents the dipole flow via
$$
\langle d_w\nu , \cdot\otimes  {\texttt{m}} \rangle=\im \mu,
$$
where the left member means that the differentials $dw$ and $d\bar{w}$ in $d_w\nu$ shall act on $\texttt{m}$. In practice this just means that
$dw$ and $d\bar{w}$ shall be replaced by $\texttt{m}$ and $\bar{\texttt{m}}$, respectively, regarded as complex numbers

However, this purely conformal picture (we have not yet used the metric) gives no information of how the dipole is to move.
That must be determined by the metric, and there seems to be no other reasonable possibility than that it shall move along 
the geodesic perpendicular to ${\texttt m}$. In particular, the regular part (\ref{finite part}) seems not to influence the dipole dynamics.
What one can say from a mathematical point of view is that the coefficient in the leading term of (\ref{finite part}),
$$
-\frac{1}{6}q_{\rm robin}(w)=\frac{\partial h_1(w)}{\partial w}-2h_2(w)
$$
is a projective connection, up to a factor, see Lemma~\ref{lem:gammaconnections}. The second term in (\ref{finite part}) is directly
linked to the two metrics involved: by (\ref{dsrobin}), (\ref{dsmetric}), (\ref{Hlambda}) we have
$$
\frac{\partial h_1(w)}{\partial \bar{w}}-2h_{11}(w)=\frac{1}{4}\Delta h_0(w)-\frac{\pi}{V}\lambda(w)^2.
$$
This entire quantity transforms as the density of a metric.
But it is not necessarily positive. In the case of a sphere, for example,
it vanishes (see Example~\ref{ex:sphere}). Also $q_{\rm robin}$ vanishes in the case.  

\begin{example}\label{rem:dipole singularity}
As an example, let the initial condition be that $w=0$ and ${\texttt m}=\frac{\dee}{\dee v} $, where $w=u+\I v$.
Then $\langle dw,{\texttt m}\rangle= \I$,
and the flow $1$-form from the  singularity becomes 
$$
 \im\frac{dz\, \langle dw,{\texttt{m}}\rangle }{(z-w)^2}=\re\frac{dz}{z^2}=\frac{(x^2-y^2)dx-2xy dy}{(x^2+y^2)^2}.
$$
On the $x$-axis  ($y=0$) this is ${dx}/{x^2}$, in other words a flow along the axis with a polar singularity of order two, as expected. 
Indeed, the entire flow is the well-known dipole picture which appears in many applications of conformal mapping.
\end{example}


\section{Remarks on vortex motion in planar domains}\label{sec:Schottky double}

Vortex motion in a planar domain can easily be treated as a special case of vortex motion on Riemann surfaces by
turning to the {\it Schottky double} of the planar domain. For simplicity we shall only discuss the case of one single vortex
in the planar domain. The case of several vortices will be similar in an obvious way. The ideas in this section extend to cases of
vortex motion on general open Riemannian surfaces with analytic boundary. 

Let $\Omega\subset \C$ be the planar domain, assumed to be bounded by finitely many real analytic curves. The Schottky double,
first described in \cite{Schottky-1877},
is the compact Riemann surface $M=\Hat{\Omega}$ obtained by completing $\Omega$ with a ``backside'' $\tilde{\Omega}$,
having the opposite conformal structure, and glueing the two along the common boundary. Thus 
$\Hat{\Omega}=\Omega \cup\partial\Omega\cup \tilde{\Omega}$ in a set theoretic sense, and the conformal structure becomes
smooth over $\partial\Omega$, as can be seen from well-known reflection principles. If $z$ is a point in $\Omega\subset M$,
then $\tilde{z}$ will denote the corresponding (reflected) point on $\tilde{\Omega}\subset M$. Both $z$ and $\tilde{z}$ can also be considered as
points in $\C$, then serving as coordinates of the mentioned points in $M$ (holomorphic respectively anti-holomorphic coordinates), and as such 
they are the same: $z=\tilde{z}\in\C$. 

In our case we need also a Riemannian structure, with a metric. This is to be the Euclidean metric on each of $\Omega$ 
and $\tilde{\Omega}$, but these do not fit smoothly across curved parts of $\partial\Omega$,
it will only be Lipschitz continuous. But that is good enough for our purposes because the vortex will
anyway never approach the boundary.  (In the case of several vortices it is however possible to make up situations in which
vortices do reach the boundary).

The metric on $M=\Hat{\Omega}$ is thus to be 
\begin{equation}\label{metricM}
ds=
\begin{cases}
|dz|, \quad z\in \Omega,\\
|d\tilde{z}|, \quad \tilde{z}\in \tilde{\Omega}.
\end{cases}
\end{equation}
To see how this behaves across
$\partial \Omega$ we need a holomorphic coordinate defined in a full neighborhood of this curve in $M$. A natural candidate can be
defined in terms of the {\it Schwarz function} $S(z)$ for $\partial\Omega$, a function which is defined by its properties of being holomorphic 
in a neighborhood of $\partial\Omega$ in $\C$ and by satisfying
\begin{equation}\label{Schwarz}
S(z)=\bar{z} \quad \text{on }\partial\Omega.
\end{equation} 
See \cite{Davis-1974, Shapiro-1992} for details about $S(z)$. We remark that 
$z\mapsto \overline{S(z)}$ is the (local) anti-conformal reflection map in $\partial\Omega$
and that $S'(z)=T(z)^{-2}$, where
$T(z)$ is the positively oriented and holomorphically extended unit tangent vector on $\partial\Omega$.

The complex coordinate $z$ in $\Omega$ extends, as a holomorphic function, to a full neighborhood of $\Omega\cup\partial\Omega$,
both when this neighborhood is considered as a subset of $\C$ and when it is considered as a subset of $M$.
The first case is trivial, and the second case depends on $\partial\Omega$ being analytic. In terms of the Schwarz function
this second extension is given by
\begin{equation}\label{extension}
z=
\begin{cases}
z\quad &\text{for }z\in\Omega\cup \partial\Omega,\\
\overline{S({\tilde{z}})}\quad&\text{for }\tilde{z}\in \tilde{\Omega},\text{ close to }\partial\Omega.
\end{cases}
\end{equation} 
In the latter expression, $\overline{S(\tilde{z})}$, $\tilde{z}$ is to be interpreted as a complex number.
When the metric on $M$ is expressed in the coordinate (\ref{extension}) it becomes
\begin{equation}\label{coordinate metric}
ds=
\begin{cases}
|dz|\quad &\text{for }z\in\Omega\cup \partial\Omega,\\
|S^\prime({z})||d{z}|\quad&\text{for }{z}\in \C\setminus\overline{\Omega},\text{ close to }\partial\Omega.
\end{cases}
\end{equation}
In the second case, $z=\overline{S(\tilde{z})}$, $\tilde{z}\in\tilde{\Omega}$, whereby 
$\tilde{z}=\overline{S(z)}$ and so $|d\tilde{z}|=|S^\prime({z})||d{z}|$. 
Thus (\ref{coordinate metric}) is consistent with (\ref{metricM}).
We see from the coordinate representation (\ref{coordinate metric}) that the metric is only Lipschitz continuous across $\partial\Omega$. 
This is the best one can expect.

The associated affine connection (\ref{rr}) (or (\ref{rlog})) is in the coordinate (\ref{extension}) given by
$$
r(z)=
\begin{cases}
0 \quad &\text{for }z\in\Omega\cup \partial\Omega,\\
\{S({z}),{z}\}_1\quad&\text{for }{z}\in \C\setminus\overline{\Omega},\text{ close to }\partial\Omega,
\end{cases}
$$
where (see Appendix~1, Section~\ref{sec:connections}, for notations)
\begin{equation}\label{rSz}
\{S(z),z\}_1=\frac{S''(z)}{S'(z)}=-2\frac{T'(z)}{T(z)}.
\end{equation}
Thus $r(z)$ is discontinuous across $\partial\Omega$ and it should on this curve be represented by its mean-value 
\begin{equation}\label{rmv}
r_{\rm MV}(z)=\frac{1}{2}\{S(z),z\}_1=-\frac{T'(z)}{T(z)} \quad (z\in\partial\Omega).
\end{equation}

\begin{example}
Let $\Omega=\D$. Then $S(z)=\frac{1}{z}$ and the coordinate $z$ in (\ref{extension}) extends to the entire complex plane, thus
representing all of $M=\D \cup \dee\D \cup\tilde{\D}$ except for the point $\tilde{0}\in\tilde{\D}$. And the metric expressed in this coordinate
becomes
$$
ds=
\begin{cases}
|dz|,\quad &|z|\leq 1,\\
|z|^{-2}{|d{z}|},\quad&|z|>1.
\end{cases}
$$
The affine connection similarly becomes, including the mean-value on the boundary, 
$$
r(z)=
\begin{cases}
0,\quad &|z|< 1,\\
-z, \quad & |z|=1,\\
-2z,\quad&|z|>1.
\end{cases}
$$
\end{example}

The geodesic curves in $\Omega$ are of course the (Euclidean) straight lines (similarly in $\tilde{\Omega}$),
geodesic curves crossing $\partial\Omega$ are straight lines reflecting into the other side under equal angles on $\partial\Omega$
(just as ordinary optical reflection),
while $\partial\Omega$ is in itself a geodesic curve. The latter is intuitively obvious since at any point on
$\partial\Omega$ there should be one geodesic in the tangential direction, and this has no other way to go
than to follow the boundary.

To confirm the last statement analytically,  let $t$ be an arc length (with respect to the Euclidean metric) 
parameter along $\partial\Omega$, so that $T(z)=\frac{dz}{dt}$ on $\partial\Omega$.
The curvature $\kappa$ of $\partial\Omega$ is
$$
\kappa= \frac{d}{dt}\arg \frac{dz}{dt} 
\quad z\in\partial\Omega,
$$
and using that $T(z)\overline{T(z)}=1$ on $\partial\Omega$ one finds that
$$
T'(z)=\I \kappa \quad (z\in \partial\Omega).
$$
In particular $T'(z)$, and so $r_{\rm MV}(z)T(z)$ (by (\ref{rmv})), is purely imaginary on $\partial\Omega$.
Combining with (\ref{rmv}) it follows that
$$
\frac{d}{dt}\arg \frac{dz}{dt}={\I}r_{\rm MV}(z)T(z),
$$
hence
$$
\frac{d}{dt}\arg \frac{dz}{dt}+{\im}\big(r_{\rm MV}(z)T(z)\big)=0 \quad (z\in\partial\Omega).
$$
Thus the geodesic equation (\ref{ddtargdwdt}) holds for the curve $\partial\Omega$, as claimed.

We remark also that the curvature $\kappa$ of the boundary curve $\partial\Omega$ appears also in the expression for the Gaussian curvature
for the metric on $M$. That curvature vanishes on $\Omega$ and on $\tilde{\Omega}$, whereas it on $\partial\Omega$
has a singular contribution, with density  $2\kappa$ with respect arc-length measure on $\partial\Omega$.

A single vortex in a planar domain $\Omega$ moves along a level line of
the appropriate Robin function, or Routh's stream function \cite{Lin-1941a, Lin-1941b, Lin-1943}).
If the vortex is close to the boundary then it follows the boundary
closely, with high speed.  From the perspective of the Schottky double 
the boundary conditions for planar fluid motion are such that there is automatically a mirror vortex on the other side in the
double, thus we really have a vortex pair close to $\partial\Omega$ on the double. In the limit this becomes a vortex dipole,
moving with infinite speed along the geodesic $\partial\Omega$.

Considering in some more detail such a symmetric vortex pair, with vortex  locations $w\in\Omega$ and $\tilde{w}\in \tilde\Omega$,
we first notice that the Green function $G^\omega(z)$ for $\omega=\delta_w-\delta_{\tilde{w}}$ simply is the anti-symmetric
extension to the Schottky double of the ordinary Green function $G_\Omega(z,w)$ for  $\Omega$:
$$
G^{\delta_w-\delta_{\tilde{w}}}(z)= G_\Omega(z,w)\quad (z\in\Omega).
$$
Then the stream function $\psi$ 
in (\ref{Ghydro}) becomes what is sometimes called the {\it hydrodynamic Green function}, which depends on the prescribed
periods. This function, which  can be traced back (at least in special cases) to \cite{Koebe-1916, Lin-1943}, 
has more recently been discussed in for example \cite{Cohn-1980, Flucher-Gustafsson-1997, Flucher-1999, Gustafsson-Sebbar-2012}.

We wish to clarify how this hydrodynamic Green function is related to the modification, in the beginning of Section~\ref{sec:hamiltonian},
of the general Green function flow $-*dG^\omega$ by an additional term $\eta$. To do this we fix, 
in the case of a Schottky double $M=\hat{\Omega}$, the homology basis in such a way that the curves $\beta_j$,
$j=1,\dots,\texttt{g}$, closely follow the inner components of $\partial\Omega$, and each curve $\alpha_j$ goes from the
$j$:th inner component of $\dee\Omega$ through $\Omega$ to the outer component, and then back again on the backside.

We also introduce the harmonic measures
$u_j$, $j=1,\dots,\texttt{g}$, here defined to be those harmonic functions in $\Omega$ which takes the boundary value one on 
the designated (number $j$)  inner component of
$\partial\Omega$ and vanishes on the rest of $\partial\Omega$. Their differentials $du_j$ extend harmonically to the Schottky double
with $\oint_{\alpha_k} du_j=-2\delta_{kj}$,  $\oint_{\alpha_k} du_j=0$. Thus $du_j=-2dU_{\beta_j}$ in terms of our general notations
(as in (\ref{taualpha}), (\ref{taubeta})).

In the block matrix notation of (\ref{RQPR}) we have
$$
-*dU_\beta=P \,dU_\alpha+R\,dU_\beta,
$$
where $P=(P_{kj})$, $R=(R_{kj})$ and (see (\ref{taualpha}), (\ref{taubeta}))
$$
P_{kj}=-\oint_{\beta_j} *dU_{\beta_k}, \quad R_{kj}=\oint_{\alpha_j}*dU_{\beta_k}.
$$
The last integral can be written 
$$
R_{kj}=-\frac{1}{2}\oint_{\alpha_j}*du_k
=-\frac{1}{2}\int_{\alpha_j\cap \Omega} \frac{\dee u_k}{\dee n}ds-\frac{1}{2}\int_{\alpha_j\cap \tilde{\Omega}} \frac{\dee u_k}{\dee n}ds,
$$
and it is easy to see that it is zero because of cancelling contributions from $\Omega$ and $\tilde{\Omega}$
due to the symmetry of $du_k$ and $\alpha_j$ going the opposite way on the backside. 

As a general conclusion we therefore have that $R=0$ in the matrix (\ref{RQPR}) when $M$ is the Schottky double of a planar domain. As a consequence,
$$
-*dU_\beta=P\,dU_\alpha.
$$
Similarly, the other equation contained in (\ref{RQPR}) becomes
\begin{equation}\label{stardUqdU}
*dU_\alpha=Q\,dU_\beta.
\end{equation}

Turning now to flow $\eta$ in (\ref{etaAB}), this is necessarily symmetric on $M=\hat{\Omega}$, hence
$$
\oint_{\alpha_j}\eta=0.
$$
It follows that the coefficients $A_j$ in (\ref{etaAB}) vanish, whereby that equation becomes
$$
\eta= \sum_{j=1}^\texttt{g}B_j dU_{\alpha_j}.
$$
In terms of the stream function $\psi=G^\omega+\psi_0$ (see again (\ref{Ghydro})) this gives, inserting also (\ref{stardUqdU}),
$$
d\psi_0=*\eta=\sum_{j=1}^\texttt{g} B_j *dU_{\alpha_j}=\sum_{j=1}^\texttt{g} C_j dU_{\beta_j}=d\big(-\frac{1}{2}\sum_{j=1}^\texttt{g} C_j u_{j} \big),
$$
with $C_j=\sum_{i=1}^{\texttt{g}} B_i Q_{ij}$. It follows in particular that $\psi_0$, and hence all of $\psi$, is single-valued on $\Omega$.

In summary, the total stream function is 
$$
\psi (z) =G_\Omega({z,w})+\sum_{j=1}^\texttt{g} C_j U_{\beta_j}(z),
$$
and it is single-valued when restricted to $\Omega$. It is clear that the $C_j$ will actually depend on $w$, 
and for symmetry reasons the above formula  eventually takes the form
$$
\psi (z) =G_\Omega({z,w})+\sum_{i,j=1}^\texttt{g} C_{ij} U_{\beta_i}(z)U_{\beta_j}(w)
$$
for some constants $C_{ij}$ subject to $C_{ij}={C_{ji}}$. 


\section{Appendix 1: Affine and projective connections}\label{sec:connections}

Besides differential forms, and tensor fields in general,
affine and projective connections are quantities on Riemann surfaces which are relevant for point vortex motion. Therefore we give below
a short introduction to these notions. The affine connections have the same meanings as in ordinary differential geometry, used to define 
covariant derivatives for example, and they play an important role in many areas of mathematical physics. 
Projective connections have more recently become important, in for example conformal field theory.

Some general references for the kind of connections we are going to consider are  
\cite{Schiffer-Hawley-1962, Gunning-1966, Gunning-1967, Gunning-1978, Gustafsson-Sebbar-2012}.  
We define them in the simplest possible manner, namely as quantities defined in terms of local holomorphic coordinates and transforming in
specified ways when changing from one coordinate to another.

Let $\tilde{z}=\varphi(z)$ represent a holomorphic local change of complex coordinate on a Riemann surface $M$ and define three nonlinear differential expressions
$\{\cdot,\cdot\}_k$, $k=0,1,2$, by
\begin{align*}
\{\tilde{z},z\}_0 &= \log \varphi'(z)&&=-2 \log \frac{1}{\sqrt{\varphi'}}\\
\{\tilde{z},z\}_1 &= (\log \varphi'(z))' =\frac{\varphi''}{\varphi'}&&= -2 \sqrt{\varphi'}\,\,(\frac{1}{\sqrt{\varphi'}})'\\
\{\tilde{z},z\}_2 &= (\log \varphi'(z))''-\frac{1}{2}((\log \varphi'(z))')^2=\frac{\varphi'''}{\varphi'}-\frac{3}{2}(\frac{\varphi''}{\varphi'})^2 &&= -2 \sqrt{\varphi'}\,\,(\frac{1}{\sqrt{\varphi'}})''
\end{align*}
The last expression is the \emph{Schwarzian derivative} of $\varphi$. For
$\{\tilde{z},z\}_0$ there is an additive indetermincy of $2\pi \I$, so
actually only its real part, or exponential, is completely well-defined.

As alternative definitions we have, with $\tilde{z}=\varphi(z)$, $\tilde{w}=\varphi(w)$, 
\begin{align*}
\{\tilde{w},w\}_0&=\lim_{z\to w} \log\frac{\tilde{z}-\tilde{w}}{z-w},\\
\{\tilde{w},w\}_1&=2\lim_{z\to w} \frac{\partial}{\partial z}\log\frac{\tilde{z}-\tilde{w}}{z-w},\\
\{\tilde{w},w\}_2&=6\lim_{z\to w} \frac{\partial^2}{\partial z\partial w}\log\frac{\tilde{z}-\tilde{w}}{z-w}.
\end{align*}

The following chain rules hold, if $z$ depends on $w$ via an intermediate variable $u$:
\begin{equation*}
\{z,w\}_k (dw)^k=\{z,u\}_k (du)^k+\{u,w\}_k (dw)^k \quad (k=0,1,2).
\end{equation*}
In particular,
\begin{equation*}
\{z,w\}_k (dw)^k=-\{w,z\}_k (dz)^k \quad (k=0,1,2).
\end{equation*}
It turns out that the three operators $\{\cdot,\cdot\}_k$, $k=0,1,2$, are unique in having properties as above,
i.e., one cannot go on with anything similar for $k\geq 3$. See \cite{Gunning-1966, Gunning-1967} for details. 

\begin{definition}\label{def:affine}
An {\it affine connection} (or $1$-{\it connection}) on $M$ is an object which is represented
by local differentials $r(z)dz$,
$\tilde{r}(\tilde{z})d\tilde{z}$,\dots (one in each coordinate
variable, and not necessarily holomorphic) glued together according to the rule
\begin{equation*}
\tilde{r}({\tilde{z}}){d\tilde{z}}=r(z){dz}-\{\tilde{z},z\}_1\,{dz}.
\end{equation*}
\end{definition}

\begin{definition}\label{def:projective}
A {\it projective connection} (or {\it Schwarzian connection}, or $2$-{\it connection}) on $M$ consists of local
quadratic differentials $q(z)(dz)^2$,
$\tilde{q}(\tilde{z})(d\tilde{z})^2$, \dots, glued together
according to
\begin{equation*}
\tilde{q}({\tilde{z}})({d\tilde{z}})^2=q(z)({dz})^2-\{\tilde{z},z\}_2\,({dz})^2.
\end{equation*}
\end{definition}

One may also consider $0$-{\it connections}, assumed here to be real-valued.
Such a connection $p(z)$ transforms according to 
\begin{equation*}
\tilde p (\tilde z)= p(z) -\re\{\tilde z, z\}_0.
\end{equation*}
This means exactly that 
\begin{equation*}
ds =e^{p(z)}{|dz|}.
\end{equation*}
is a Riemannian metric.

For a metric in general, $ds =\lambda(z){|dz|}=e^{p(z)}|dz|$, there is a natural affine connection associated  to it by
\begin{equation}\label{rlog}
r(z)= 2 \frac{\partial }{\partial z}\log \lambda(z)=2\frac{\partial p}{\partial z}
=\frac{\partial p}{\partial x}-\I \frac{\partial p}{\partial y}.
\end{equation}
This can be identified with the Levi-Civita connection in general tensor analysis.
The real and imaginary parts coincide (up to sign) with the components 
of the classical Christoffel symbols $\Gamma_{ij}^k$.  The Gaussian curvature of the metric is 
$$
\kappa(z)=-4\lambda(z)^{-2}\frac{\partial^2}{\partial z\partial \bar{z}}\log\lambda(z)=-2{\lambda(z)^{-2}}\frac{\partial r(z)}{\partial\bar{z}}.
$$
Under conformal changes of coordinates $\kappa$ transforms as a scalar: $\tilde{\kappa}(\tilde{z})=\kappa (z)$ in previous notation.

Independent of any metric, an affine connection $r$ gives rise to a projective connection $q$ by
\begin{equation}\label{qr}
q(z)= \frac{\partial r(z)}{\partial z}- \frac{1}{2}r(z)^2.
\end{equation}
This $q$ is sometimes called the ``curvature" of $r$ (see \cite{Dubrovin-1993}). That curvature is however not the same as the Gaussian curvature in case
$r(z)$ comes form a metric. A polarized version of $q$ is
\begin{equation}\label{qpolarized}
q(z,w)= \frac{1}{2}\big(\frac{\partial r(z)}{\partial z}+\frac{\partial r(w)}{\partial w}-r(z)r(w)\big),
\end{equation}
related to the Batman function in \cite{Boatto-Koiller-2013}.

The equation for geodesic curves $z=z(t)$ is, in terms of an arc-length parameter $t$,
$$
\frac{d^2 z}{dt^2 }+r(z)(\frac{dz}{dt})^2=0
$$
or, written in another way,
\begin{equation}\label{geodesics}
\frac{d}{dt}\log \frac{dz}{dt}+r(z)\frac{dz}{dt}=0.
\end{equation}
The first version is just a reformulation of the usual equation in terms of Christoffel functions in ordinary differential geometry (see \cite{Frankel-2012}).

The real part of (\ref{geodesics}) only contains information about the parametrization, while the imaginary part, namely
\begin{equation}\label{imgeodesics}
\frac{d}{dt}\arg \frac{dz}{dt}+\im\big( r(z)\frac{dz}{dt}\big)=0.
\end{equation}
describes the geodesic curve in terms of an arbitrary real parameter $t$. The latter property is useful in the context of the motion
of vortex pairs since the speed of such a pair tends to infinity as the distance between the two vortices goes zero, and time therefore 
has to be successively reparametrized.


\section{Appendix 2: Behavior of singular parts under changes of coordinates}\label{sec:singular parts} 

The stream function $\psi$  for a flow in a neighborhood of a point vortex with location $w$
has an expansion starting, in local coordinates and up to a constant factor depending on the strength of the vortex,
$$
\psi(z)=\log|z-w|-c_0(w) -\mathcal{O}(|z-w|).
$$
Similarly, the expansion of the analytic completion $\nu+\I*\nu=2\I\frac{\partial\psi}{\partial{z}}dz$ of the flow $1$-form $\nu$ 
is (again up to a constant factor)
$$
f(z)dz=\frac{dz}{z-w} -c_1(w)dz -\mathcal{O}(z-w).
$$ 

One step further, one may consider a pure vortex dipole. The corresponding flow is
obtained by differentiating the above $\nu=\re \big(f(z)dz\big)$  with respect to $a$.
The analytic completion $F(z)dz dw$ of the so obtained flow $1$-form $d_w\nu$ has an expansion 
$$
F(z)dz dw= \frac{dz dw}{(z-w)^2}-2c_2(w)dz dw + \mathcal{O}(z-w).
$$
Above, and always in similar situations, $dzdw$ should be interpreted as the tensor product $dz\otimes dw$, i.e. it is not a wedge product. 

The coefficients $c_0$, $c_1$, $c_2$ above transforms, under conformal changes of coordinates, as different kinds of connections
as defined in Appendix, Section~\ref{sec:connections}. 
This is made precise in the following elementary lemma, which we here cite without proof from \cite{Gustafsson-2019}.
Compare Lemma~3.6 and subsequent discussions in \cite{Biswas-Raina-1996}.

\begin{lemma}\label{lem:connections}
Let $\tilde{z}=\varphi(z)$ be a local conformal map representing a change of coordinates near a point $z=w$, and set $\tilde{w}=\varphi (w)$.
Let $\psi(z)$ ($=\psi(z,w)$) be a locally defined real-valued harmonic function with a logarithmic pole at $z=w$,  
similarly $f(z)dz$ a meromorphic differential with a simple pole 
at the same point, and $F(z)dzdw$ a double differential with a pure (residue free) second order pole.
Precisely, we assume the following local forms, in the $z$ and $\tilde{z}$ variables:
\begin{align*}
\psi(z)&= \log|z-w|-c_0(w) + \mathcal{O}(|{z}-{w}|)\\
&=\log|\tilde{z}-\tilde{w}|-\tilde{c}_0(\tilde{w})+ \mathcal{O}(|\tilde{z}-\tilde{w}|),\\
f(z)dz&= \frac{dz}{z-w}-c_1(w)dz + \mathcal{O}(z-w)\\
&=\frac{d\tilde{z}}{\tilde{z}-\tilde{w}}-\tilde{c}_1(\tilde{w})d\tilde{z}+ \mathcal{O}(\tilde{z}-\tilde{w}),\\
F(z)dzdw&= \frac{dzdw}{(z-w)^2}-2c_2(w)dzdw + \mathcal{O}(z-w)\\
&=\frac{d\tilde{z}d\tilde{w}}{(\tilde{z}-\tilde{w})^2}-2\tilde{c}_2(\tilde{w})d\tilde{z}d\tilde{w}+ \mathcal{O}(\tilde{z}-\tilde{w}).
\end{align*}
Then, as functions of the location of the singularity and up to constant factors, $c_0$ transforms as a $0$-connection,
$c_1$ as an affine connection and $c_2$ as a projective connection:
$$
\tilde{c}_0(\tilde{w})= c_0(w)+\re\{\tilde{w},w\}_0, 
$$
$$
\tilde{c}_1(\tilde{w}) d\tilde{w}= c_1 (w)dw +\frac{1}{2} \{\tilde{w}, w\}_1 dw,
$$
$$
2\tilde{c}_2(\tilde{w}) d\tilde{w}^2= 2c_2 (w)dw^2 +\frac{1}{6} \{\tilde{w}, w\}_2 dw^2.
$$
The first statement is most conveniently expressed by saying that
$$
ds=  e^{-\tilde{c}_0(\tilde{w})}|d\tilde{w}|=e^{-c_0(w)}|dw|
$$
defines a conformally invariant metric.
\end{lemma}

\begin{remark}
It is actually not necessary for the conclusions of the lemma that $\psi$, $fdz$, $Fdzdw$ are harmonic/analytic away from the singularity,
it is enough that the local forms of the singularity and constant terms given as above hold.
\end{remark}

We may adapt the above lemma to the Green function $G(z,w)=G^{\delta_w} (z)$, despite it  is not harmonic in $z$
(because of the compensating background flow).  First we write
\begin{equation}\label{GlogH1}
G(z,w)=\frac{1}{2\pi} (-\log |z-w| +H(z,w))
\end{equation}
and then expand the regular part as
\begin{equation}\label{greentaylor1}
H(z,w)= h_0 (w)+\frac{1}{2}\left(h_1(w)(z-w)+ \overline{h_1 (w)} (\bar{z}-\bar{w})\right)+
\end{equation}
$$
+\frac{1}{2}\left(h_{2}(w)(z-w)^2+\overline{h_{2}(w)}(\bar{z}-\bar{w})^2\right)+h_{11}(w)(z-w)(\bar{z}-\bar{w})+ \mathcal{O}(|z-w|^3).
$$

We note from (\ref{greentaylor}) that
$$
H(w,w)=h_0(w),\quad
\{\frac{\partial H(z,w)}{\partial z} \}_{z=w} =\frac{1}{2} h_{1}(w).
$$
Thus the symmetry of $H(z,w)$ gives
\begin{equation}\label{c1c0}
h_1(w)=\frac{\partial h_0 (w)}{\partial w},
\end{equation}
For the second order derivatives we have
\begin{equation}\label{Hlambdaz}
\{\frac{\partial^2 H(z,w)}{\partial z^2}\}_{z=w} = h_{2}(w),\qquad\qquad\quad
\end{equation}
\begin{equation}\label{Hlambda}
\{\frac{\partial^2 H(z,w)}{\partial z \partial\bar{z}}\}_{z=w} = h_{11}(w)=\frac{\pi}{2 V}\lambda(w)^2,
\end{equation}
\begin{equation}\label{Hh2}
\{\frac{\partial^2 H(z,w)}{\partial z \partial w}\}_{z=w} =\frac{1}{2}\frac{\partial h_1(w)}{\partial w} -h_{2}(w),
\end{equation}
\begin{equation}\label{Hh11}
\{\frac{\partial^2 H(z,w)}{\partial z \partial\bar{w}}\}_{z=w} = \frac{1}{2}\frac{\partial h_1(w)}{\partial \bar{w}} -h_{11}(w).
\end{equation}

The transformation properties of the coefficients $h_j(w)$ are closely related
to those appearing in Lemma~\ref{lem:connections}. Citing from \cite{Gustafsson-2019} we have 
\begin{lemma}\label{lem:gammaconnections}
Under a local holomorphic change 
$\tilde{z}=\varphi(z)$ of coordinates, with $\tilde{w}=\varphi (w)$, we have
\begin{equation}\label{h0}
\tilde{h}_0(\tilde{w})= h_0(w)+\re\{\tilde{w},w\}_0, 
\end{equation}
\begin{equation}\label{h1}
\tilde{h}_1(\tilde{w}) d\tilde{w}= h_1 (w)dw +\frac{1}{2} \{\tilde{w}, w\}_1 dw,
\end{equation}
\begin{equation}\label{h2h1}
\big(\frac{\partial \tilde{h}_1(\tilde{w})}{\partial \tilde{w}}-2\tilde{h}_2(\tilde{w}) \big)d\tilde{w}^2
= \big(\frac{\partial h_1(w)}{\partial w}-2{h}_2(w) \big)d{w}^2 +\frac{1}{6} \{\tilde{w}, w\}_2 dw^2,
\end{equation}
\begin{equation}\label{h11}
\tilde{h}_{11}(\tilde{w})|d\tilde{w}|^2=h_{11}(w)|dw|^2.
\end{equation}
\end{lemma}


\bibliography{bibliography_gbjorn.bib}

\end{document}